\documentclass[reqno]{amsart}
\usepackage{amssymb}
\usepackage{pdfsync}
\usepackage{hyperref}
\usepackage{color}
\usepackage{graphics,pstricks}
\usepackage{tikz}

\newtheorem{theorem}{Theorem}[section]
\newtheorem{proposition}[theorem]{Proposition}
\newtheorem{lemma}[theorem]{Lemma}

\newtheorem{definition}[theorem]{Definition}
\newtheorem{remark}[theorem]{Remark}

\numberwithin{equation}{section}
\newenvironment{acknowledgement}{\emph{Acknowledgement.}}

\DeclareMathOperator{\dv}{div}

\DeclareMathOperator{\tr}{tr}

\DeclareMathOperator{\im}{Im}

\newcommand{\RR}{\mathcal{R}}
\newcommand{\B}{\mathcal{B}}
\newcommand{\I}{\mathcal{I}}

\newcommand{\A}{\mathcal{A}}
\newcommand{\eps}{\epsilon}
\newcommand{\h}{\mathcal{H}}
\newcommand{\R}{\mathbb{R}}

\newcommand{\ee}{\mathsf{E}}

\newcommand{\z}{\mathbb{Z}}
\newcommand{\N}{\mathbb{N}}
\newcommand{\M}{\mathcal{M}}
\newcommand{\m}{\mathsf{M}}
\newcommand{\Es}{\mathbb{E}}
\newcommand{\T}{\mathcal{T}}
\newcommand{\cc}{\mathcal{C}}
\newcommand{\w}{\mathcal{W}}

\newcommand{\proba}{\mathbb{P}}
\newcommand{\LL}{\Lambda}
\newcommand{\p}{P_\omega^{(E)}}
\newcommand{\pp}{P_\omega^{(E) \bot}}
\newcommand{\la}{\langle}
\newcommand{\ra}{\rangle}

\newcommand{\e}{\mathrm{e}}
\newcommand{\dd}{\mathrm{d}}

\newcommand{\com}[1]{\left[#1\right]}
\newcommand{\pnorm}[1]{{\left\lVert #1 \right\rVert}_p}
\newcommand{\abs}[1]{\left\lvert #1 \right\rvert}
\newcommand{\norm}[1]{\left\lVert #1 \right\rVert}
\newcommand{\trnorm}[1]{{\left\lVert #1 \right\rVert}_1}
\newcommand{\normsch}[1]{{\left\lVert #1 \right\rVert}_2}
\newcommand{\scal}[1]{\la #1 \ra}
\newcommand{\f}{\mathcal{X}}
\newcommand{\style}{\displaystyle}

\begin{document}

\title[Bulk-edge equality of conductances]{Equality of bulk and edge Hall conductances for continuous magnetic random Schr\"odinger operators}
\author{Amal Taarabt}
\address{ Universit\'e de Toulon, CNRS, CPT, UMR 7332, 83957 La Garde, France}
\email{amal.taarabt@univ-tln.fr}


\begin{abstract} 
In this note, we prove the equality of the quantum bulk and the edge Hall conductances in mobility edges and in presence of disorder. 
The bulk and edge perturbations can be either of electric or magnetic nature. 
The edge conductance is regularized in a suitable way to enable the Fermi level to lie in a region of localized states.
\end{abstract}

\maketitle
\tableofcontents

\section{Introduction}\label{intro}

A large literature has emerged a few years only after the discovry of the integer quantum Hall effect \cite{KDP}.
Laughlin followed by Halperin, argued that the occurence of the \textit{plateaux} is due to the localization phenomenon \cite{Hal,L}.
The presence of impurities is indeed imperative in order to observe the quantum Hall effect \cite{B,BESB}.
In a disorder media, the energy spectrum consists in bands of extended states separated by energy regions of localized states or energy gaps 
\cite{BESB,GKS}.
If the Fermi energy lies in the extremities of these bands, where localization holds, the Hall conductance is constant. The quantum Hall conductance jumps from one integer value to another 
near the centers until to find a new localization region.
Halperin formulates the existence of the edge currents in the Hall systems \cite{Hal}. 
Indeed, the electrons flowing along the edge of the system rebond and induce then currents which are quantized through the edge conductance. 
He established that these conductances are \textit{a priori} equal.

The mathematical study of the quantization of Hall conductances has been first developed in parallel. While Bellissard and followers \cite{B,BESB,ASS,BoGKS,GKS1,GKS2} were 
interested in the Hall conductance, its topological nature, its quantization, and its derivation from a Kubo formula of the quantum Hall effect which is a part 
of the theory of noncommutative geometry, \cite{CG,CGH,DGR1,KSB,KRSB1,KRSB2} are rather devoted to the edge currents and their quantization. 
These similtanuous quantizations highlight the equality of the edge and bulk conductances that \cite{EG, EGS} showed by derivation in the discrete case. 
Elbau and Graf showed that the bulk and the edge conductances matches and are equal under a gap condition \cite{EG}. It was later improved in \cite{EGS} for energy
intervals lying in localization region of the bulk Hamiltonians.
Our goal in this paper is to prove that equality within the context of random magnetic Schr\"odinger operators in the continuum and in presence of electric or magnetic wall.

A great interest has been focused in the recent years on the study of random magnetic fields and their localisation properties \cite{AHK,BSK,CH,CHKR,DGR2,GhHK, W}.
To describe the bulk in our model, we consider electric and magnetic random perturbations of the free Landau Hamiltonian of Anderson type.
The associated Hall conductance is stemmed from the Kubo formula.
We then introduce a confining wall, that will be sent to infinity.
The models that we deal with are purely electric or purely magnetic (wall and random pertubation). However, we could also consider variants
with an electric random operator and a magnetic wall and vice-versa.
We define the associated operators and edge conductance.

It is important to emphasize that a localization regime requires a regularization of the usual edge conductance.
These regularizations are intended to cancel the contributions of states living away from the edge that might generate extra currents
and to restore the trace class property which could be destroyed.
We shall make use a regularization introduced in \cite{EGS}, and establish the equality of the bulk and edge Hall conductances  
by deriving one from the other, and not by separate quantization as in \cite{CG}.

The paper is organized as follow. In section~\ref{Bulk}, we introduce our bulk models and formulate the localization assumption. 
The section~\ref{Edge} is devoted to the description of models with electric or magnetic walls and the associated edge conductance.
The section\ref{results}, we state our main result and we sketch the strategy of its proof.
In section~\ref{proofs}, we provide the full proofs of the key steps described in section~\ref{results}.
Appendix~\ref{A} and Appendix~\ref{s-cv} contain some technical tools and trace-class properties.

\section{Bulk models}\label{Bulk}

We consider the Landau Hamiltonian 
\begin{equation}\label{landau} 
H_B=(-i\nabla-\A_0)^2 \quad{\mathrm{with}} \ \ \A_0(x_1,x_2)=\frac B2(-x_2,x_1), 
\end{equation}
where $\A_0$ is the vector potential generating a magnetic field with strength $B>0$ constant.
We shall consider electric and magnetic perturbations $V$ and $A$ of $H_B$ and we set $$H_B(A,V):=(-i\nabla-\A_0-A)^2 +V.$$
\newline
We recall the Leinfelder-Simader conditions (LS) for an operator of the form  
\begin{equation}\label{H(A,V)}
H(A,V):=(-i\nabla-A)^2+V.
 \end{equation}
We say that the magnetic potential $A$ and the electric potential $V$ satisfy the Leinfelder-Simader conditions if
\begin{enumerate}
 \item $A\in\mathrm{L}_\mathrm{loc}^2(\R^2,\R^2)$ avec $\dv A\in\mathrm{L}_{\mathrm{loc}}^2(\R^2,\R)$.
 \item $V(x)=V_+(x)-V_-(x)$ with $V_\pm\ge0, V_\pm\in\mathrm{L}_\mathrm{loc}^2(\R^2,\R)$ and $V_-$ relatively bounded with respect to 
$-\Delta$ with relative bound $<1$ such that there exist $0\le\alpha<1$ and $\beta\ge0$ independent of $\omega$ so that for all $\psi\in\mathcal{D}(\Delta)$, we have
\begin{equation*}
 \norm{V_-\psi}\le\alpha\norm{\Delta\psi}+\beta\norm{\psi}.
\end{equation*}
\end{enumerate}
Under these conditions, the operator $H(A,V)$ in \eqref{H(A,V)} is essentially self-adjoint on $\mathcal{C}_c^\infty(\R^2)$ \cite{LS}. 
\bigskip

In this work, we are interested in random perturbations of $H_B$ of electric and magnetic nature.

\subsection{Electric model}

We consider the random Landau Hamiltonian 
\begin{equation}\label{Bulk H elec} H_\omega^\ee=H_B + V_\omega \ \quad{\mathrm{on}} \ \ L^2(\R^2),\end{equation}
with $V_\omega$ a random potential of Anderson-type
\begin{equation}\label{V_omega}
 V_\omega :=\sum_{\gamma\in\z^2} \omega_\gamma u(.-\gamma),
\end{equation}
where $\omega=(\omega_\gamma)_{\gamma\in\z^2}$ is a family of independant and identically
distributed (iid) random variables
and the single site potential $u$ is a nonnegative bounded measurable function on $\R^2$ with compact support 
such that $-M_1\le V_\omega \le M_2$ with $0\le M_1,M_2<\infty$. 
We assume that the family $(\omega_\gamma)_\gamma$ has a common non-degenerate
probability distribution $\mu$ with bounded density $\rho$.
We write $(\Omega,\mathbb{P})$ for the underlying probability space and $\Es$ for the corresponding expectation. 
\bigskip

Using the magnetic translation $U_\alpha$ defined by
\begin{equation}\label{magn transl} 
 (U_\alpha \psi)(x)=\e^{-i\frac{B}{2}\alpha\wedge x}\psi(x-\alpha) \ \ \mathrm{for} \ \alpha\in\R^2,
\end{equation}  
where $\alpha\wedge x=\alpha_1 x_2-\alpha_2 x_1 $,
it follows that the random operator $H_\omega$ is $\z^2$-ergodic and is essentially self-adjoint with core $\mathcal{C}_c^\infty(\R^2)$. 
Hence, it follows from \cite{CL} that $H_\omega^\ee$ has a nonrandom spectrum and there exists a deterministic set $\Sigma_\ee\subset\R$ such that
$\sigma(H_\omega^\ee)=\Sigma_\ee$ with probability one.
\bigskip

The spectrum of the free Landau Hamiltonian $H_B$ given in \eqref{landau} consists in a sequence of infinitely degenerated eigenvalues, 
called Landau levels 
\begin{equation} \label{Landau levels} B_n=(2n-1)B, \ n=1,2,\dots \end{equation}
with the convention $B_0=-\infty$. And we have 
\begin{equation}\Sigma_\ee\subset\bigcup_n [B_n-M_1, B_n+M_2],\end{equation}
and there is no overlap provided that the open gap condition $M_1+M_2<2B$ is fulfilled. 

\begin{remark}
The assumption on the bounded density $\rho$ of the random variables is made in order to cover models that are known to exhibit dynamical localisation. 
\end{remark}

\subsection{Pure magnetic model}\label{magn mod}
Let $\A_\omega$ be a random vector potential of the form 
$$\style \A_\omega = \sum_{\gamma\in\z^2} \omega_\gamma u(.-\gamma),$$
satisfaying the Leinfelder-Simader conditions \cite{LS} almost surely.
The single site functions $u=(u_1,u_2)\in \cc^1(\R^2,\R^2)$ are compactly supported and
the random variables $\omega=(\omega_\gamma)_{\gamma\in\z^2}$ are independant and identically distributed (iid) with common probability distribution. 
The probability space is again denoted by $(\Omega,\proba)$.
We consider the magnetic random operator
\begin{equation}\label{Bulk H magn}
  H_\omega^\m= (-i\nabla-\A_0-\A_\omega)^2 \ \quad{\mathrm{on}} \ \ L^2(\R^2),
 \end{equation}
which is essentially self-adjoint on $\mathcal{C}_c^\infty(\R^2)$ and uniformly bounded from below
for $\proba-\mathrm{a.e} \ \omega$.
By ergodicity, we denote its spectrum $\sigma(H_\omega^\m)$ by $\Sigma_\m$. 
\\

The operators $H_\omega^\ee$ and $H_\omega^\m$ are essentially self-adjoint
and bounded from below:
there exists $\Theta_\bullet\ge 1$ such that $H_\omega^\bullet +\Theta_\bullet\ge 1$ \cite{BoGKS}
and where $\bullet=\ee,\m$.
For simplification and since our analysis remains essentially the same for both models, we may omit $\ee$ and $\m$ from the notations and write 
$H_\omega$ to denote $H_\omega^\ee$ and $H_\omega^\m$. 
Neverthless, when needed, we will specify the case we deal with.
\subsection{Localisation}

For $m>0$ and $\zeta\in (0,1]$ given, we introduce the random $(m,\zeta)$-subexponential moment at time $t$ for the time evolution, 
initially localized around the origin and localized in energy by the function $\f\in\mathcal{C}_{c,+}^\infty(I)$,
\begin{align}\label{moment} 
M_{\omega}(m,\zeta,\f,t):= &\normsch{\e^{{\frac{m}{2}}\abs{X}^{\zeta}} \e^{-itH_\omega} \f(H_\omega) \chi_0}^2. 
\end{align}
We define its time average expectation as 
\begin{equation}\label{avrg mom}
\M(m,\zeta,\f,T) := \frac{1}{T}\int_{0}^{\infty} \e^{-t/T} \ \Es\{ M_{\omega}(m,\zeta,\f,t)\} \ \dd t.
\end{equation}
Given an energy $E\in\R$, we consider the Fermi projector $\p=\chi_{(-\infty,E]}(H_\omega)$, the spectral projection of 
$H_\omega$ onto energies below $E$.

\begin{definition}\mbox{}
\begin{enumerate}
\item[{\bf(Loc)}]We say that the operator $H_\omega$ exhibits localization in an open interval $I$ if
there exist $m>0,\zeta\in (0,1)$ so that for any $\f\in C_{c,+}^{\infty}(I)$, we have

\begin{equation}\label{loc} \sup_{T} \M(m,\zeta,\f,T)<\infty.\end{equation}
We denote by $\varSigma_{loc} $ the region of localization 
\begin{equation}\label{loc region}
 \varSigma_{loc} :=\{ E\in\R: H_\omega \ \mathrm{exhibits \ localization \ in \ a \ neighborhood \ of} \ E\}.
\end{equation}

\item[{\bf(DFP)}]The Fermi projection $\p$ exhibits sub-exponential decay if the Fermi energy $E\in\varSigma_{loc}$ and if there exist 
$m>0,\zeta\in (0,1)$ such that we have
\begin{equation}\label{DFP_exp}
 \Es\left\{\normsch{\chi_x \p \chi_y}^2\right\}\leq C_{m,\zeta,B,E} \ \e^{-m|x-y|^\zeta} \ \ \mathrm{for \ all} 
\ x,y\in\z^2,
\end{equation}
where the constant $C_{m,\zeta,B,E}$ is locally bounded in $E$.
As a consequence, for any $\eps>0$ and $\proba$-a.e $\omega$, we have
\begin{equation}\label{DFP_omega}
\normsch{\chi_x \p \chi_y} \leq C_{\omega,m,\zeta,\eps,B,E} \ \e^{\eps|x|^\zeta} \e^{-m|x-y|^\zeta} \ \ \mathrm{for \ all} \ x,y\in\z^2.
\end{equation}
\end{enumerate}
\end{definition}

The existence of the region of localization \eqref{loc region} has been proven in \cite{CH,GK2,DGR2}.
Moreover, it corresponds to the region where the bootstrap multiscale analysis \textit{(MSA)} can be performed \cite{GK1,GK2}.
The magnetic models are traited in \cite{DGR2,GhHK}.
The {\bf (DFP)} property and \eqref{DFP_omega} play an important role in the study and the definition of the Hall conductance.

\subsection{Hall conductance}

Consider a smooth characteristic function $ \LL(s)$ which is equal to $1$ for $s\le-\frac12$ and $0$ for $s\ge\frac12$ such that 
$\mathrm{supp} \ \LL'\subset(-\frac12,\frac12)$. Let $\LL_j$ denotes the multiplication operator by the function $\LL_j(x)=\LL(x_j)$ for $j=1,2$.

\begin{definition}
The Hall conductance at a Fermi energy $E$ is defined by 

\begin{equation}\label{Hall conductane}
 \sigma_{\mathrm{Hall}}(B,\omega,E):= -i\tr\com{P_\omega^{(E)} \LL_2 \p, \p \LL_1 \p}. 
\end{equation}
\end{definition}
In view of \eqref{DFP_exp}, it is well defined in $\varSigma_{loc}$ (see \cite{GKS1}). 
The ergodicity property implies that \eqref{Hall conductane} 
is a nonrandom quantity in the sense that for $\mathbb{P}$-a.e $\omega$,
 \begin{equation}
  \sigma_{\mathrm{Hall}}(B,E):=\Es\left\{\sigma_{\mathrm{Hall}}(B,\omega,E)\right\}=\sigma_{\mathrm{Hall}}(B,\omega,E). 
 \end{equation}
Notice that the operators $\p\LL_2 \p$ and  $\p\LL_1 \p$ in \eqref{Hall conductane} are not separatly trace class otherwise the commutator would be zero.
\bigskip

The Hall conductance $\sigma_\mathrm{Hall}(B,E)$ is known to be constant in $\varSigma_{loc}$ \cite{GKS}.
This corresponds to the occurence of the well-known plateaux in the QHE.

\begin{remark}
 There are alternative definitions to \eqref{Hall conductane}, namely 
\begin{equation}\label{alt def 1}
-i\tr{\p \com{ \com{\p,\LL_2},\com{\p,\LL_1}}}.
\end{equation}
Note that the operator $\com{\p,\LL_2}\com{\p,\LL_1} $ in \eqref{alt def 1} is morally supported near the origin.
One can also consider
\begin{equation}\label{alt def 2}
  -i\tr\{\chi_0 \ \p \com{\com{\p,X_2},\com{\p,X_1}}\chi_0\},
\end{equation}
where $X_i$ is the multiplication operator by the coordinate $x_i$ for $i=1,2$.
\end{remark}

\section{Models with walls}\label{Edge}

In this note, we are interested in soft walls of magnetic or electric nature.

\subsection{Electric edge}

Let  $U\in\mathcal{C^{\infty}}(\R^2)$  be an $x_2$-invariant decreasing function
such that 
\begin{equation}\label{assump U}
\begin{cases} 
\style\lim_{x_1 \to-\infty} U(x_1)=U_-<\infty,\\
 U(x_1)=0 \ \ \ \ \ \mbox{for} \ \ \ x_1\ge0.
\end{cases}
\end{equation}
We should consider $U_-$ sufficiently large compared to the energy zone where we work.
The electric edge operator is giving by 

\begin{equation}\label{elec edge} 
H_{\omega,a}^\ee=H_B + U_a +V_\omega, 
\end{equation}
where $a>0$ and $U_a$ is the multiplication  
by the function $U_a (x_1)=U(x_1+a)$ which translate the wall and placing it at $x_1=-a$. It is a soft and left confining wall in the sense
that the particle remains trapped and confined on the right side of the plane.
\subsection{Magnetic edge}
Let $\A=(\A_1,\A_2)$ be a vector potential generating the magnetic field $\B:\R^2 \to \R$, i.e,
\begin{equation}
 \nabla\wedge  \A(x)=\B(x), \ \mathrm{for} \ x=(x_1,x_2)\in\R^2,
\end{equation}
where $\B$ is a smooth decreasing $x_2$-invariant function so that 
\begin{equation}\label{assump B}
\begin{cases} 
\style\lim_{x_1 \to-\infty} \B(x_1)=B_-<\infty,\\
\B(x_1)=0 \ \ \ \mbox{for} \ \ x_1\ge1. \end{cases}
\end{equation}
Once again, like the electric case abose, we translate this wall with a parameter $a>0$ so that
\begin{equation}
 \frac{\partial \A_2}{\partial x_1}(x_1+a,x_2) - \frac{\partial \A_1}{\partial x_2}(x_1+a,x_2)=\B(x_1+a):=\B_a(x_1).
\end{equation}
In that case, the Magnetic edge operator is  

\begin{equation}\label{magn edge}
 H_{\omega,a}^\m=(-i\nabla-\A_0-\A_a-\A_\omega)^2.
\end{equation}
If we set $\A_a^{\mathrm{Iw}}=\A_0 + \A_a$, we obtain the so-called Iwatsuka magnetic field with limits in $+\infty$ and $-\infty$ given by $B+B_-$ and $B$
respectively \cite{CFKS,DGR1,E,I}.

\bigskip

In view of the gauge invariance for magnetic operators, one can choose a suitable transformation and simplify the spectral studies of
magnetic operators of the form $(-i\nabla-\A)^2$. Let us consider the Laudau gauge and take $\A=(0,\A_2)$ where 
$\A_2=\beta(x_1):=\int_0^{x_1} \B(s) \dd s$. The invariance in $x_2$-direction allows the performance of the partial Fourier 
transform with respect to the variable $x_2$. Hence, the operator $H(\A)$ can be written as 
\begin{equation}H(\A)=-\frac{\partial^2}{{\partial x_1}^2} + (-i\frac{\partial}{\partial x_2} -\beta(x_1))^2. \end{equation}
 Then it is unitary equivalent to
 \begin{equation}
  h(k):= -\frac{\dd^2}{{\dd x_1}^2} + (k -\beta(x_1))^2 , \ \mathrm{for} \ k\in\R,
 \end{equation}
whose spectrum is discrete \cite{E}.

 \bigskip
 
By $H_{\omega,a}$, we mean both $H_{\omega,a}^\ee$ and $H_{\omega,a}^\m$.
Notice that the edge operators $H_{\omega,a}$ converge to $H_\omega$ in strong resolvent sense.
Hence $H_{\omega,a}\to H_\omega$ in the strong resolvent sense (see appendix~\ref{strong cv R}).
In order to justify this strong convergence, we the resolvent identity and we consider the difference operator 
\begin{equation}
 \Gamma_{\omega,a}^\m=H_{\omega,a}^\m-H_\omega^\m= -2 \A_a.(-i\nabla-\A_0-\A_\omega)+i\dv \A_a+|\A_a|^2,
\end{equation} 
and 
\begin{equation}
 \Gamma_{\omega,a}^\ee=H_{\omega,a}^\ee-H_\omega^\ee= U_a.
\end{equation} 
Since the operator $\Gamma_{\omega,a}^\bullet R_{\omega,a}$ is uniformely bounded in $a$ for $\bullet=\ee,\m$ and the compactly supported functions are dense in $\h$,
it suffices to verify this strong convergence in $\mathcal{C}_0^\infty(\R^2) $. We consider a test function 
$\phi\in\mathcal{C}_0^\infty(\R^2) $ leaving far apart from the wall such that $\mathrm{supp}\ \phi \cap \mathrm{supp}\ \B_a=\emptyset$, according to \cite{T}.

\subsection{Edge conductance}
We start with the definition of switch functions.

\begin{definition}\label{switch function}
 Let $g:\R\rightarrow[0,1]$ be a smooth decreasing function. We say that $g$ is a switch function if it has a compactly supported
derivative such that $g\equiv 1$ on the left side of $\mathrm{supp} \ g'$ and $g\equiv 0$ on the right one.
\newline
We say that $g$ is a switch function of an interval $I$ if $\mathrm{supp}\ g'\subset I$.
\end{definition}

Heuristically, the current along $x_1=-a$ and in direction $x_2$ induced by states with energy support in an interval $I$, is given by
\begin{equation*}
 J(I)=\tr(E_I(H_{\omega,a}) i\com{H_{\omega,a},\LL_2}),
\end{equation*} 
where $E_I(H_{\omega,a})$ is the spectral projection of $H_{\omega,a}$ on $I$. 
The edge conductance is then the ratio

\begin{equation*}
 \sigma_e(\omega,I)=\frac{J(I)}{|I|}\approx-i\tr(g'(H_{\omega,a})\com{H_{\omega,a},\LL_2}),
\end{equation*}
since
$$\style-g'(H_{\omega,a})\approx\frac{E_I(H_{\omega,a})}{|I|},$$
where $I$ lies in a spectral gap of $H_{\omega}$. However, it is more relevant for  physical interest, 
to consider the case where $\I$ falls into $\varSigma_{loc}$, region of localized states so that $I\cap\varSigma_{loc}\ne\emptyset$.
In fact, such states might generate spurious currents that we have to cancel.
In order to treat this case, a regularization of the edge conductance is required.
Some regularizations have been proposed in \cite{CG,CGH} and \cite{EGS}. The second candidate of \cite{EGS} is a time-average regularization where they considered
the Heinsenberg evolution of $\LL_1$ and time-averaged the final expression.
It is the regularization that we shall consider.
\bigskip

\begin{definition}\label{def g}
Let $I\subset (B_n,B_{n+1}) \cap \varSigma_{\mathrm{loc}}$ be a given interval for some $n$.
Let $g$ be a decreasing switch function of $I$. The regularized edge conductance of $H_\omega$ in $I$ is 
defined as
\begin{equation}\label{edge_cond}
\sigma_{e,\omega}^\mathrm{reg}:=\lim_{T\to\infty}\lim_{a\to\infty} \frac{1}{T} \int_0^T -i\tr g'(H_{\omega,a})
\com{H_{\omega,a},\LL_2}\LL_{1,a}^\omega(t)\dd t,
\end{equation}
whenever the limits exist and where $\LL_{1,a}^\omega(t):=\e^{itH_{\omega,a}}\LL_1 \ \e^{-itH_{\omega,a}}$. 
\end{definition}

Since the operator $g'(H_{\omega,a})\com{H_{\omega,a},\LL_2}\LL_{1,a}^\omega(t)$ is bounded, we only have to verify that the trace in \eqref{edge_cond} is well defined
and that such limits exist. The idea relies on the fact that far from the edge, the dynamic of $ \LL_{1,a}^\omega$ approaches that of 
$ \e^{itH_\omega}\LL_1\e^{-itH_\omega} $.

\begin{remark}
 Notice that both definitions \eqref{Hall conductane} and \eqref{edge_cond} do not depend either on $g$ as long as $\mathrm{supp}g'\subset I$ nor on $\LL_j$ for $j=1,2$.
\end{remark}

\section{Main result}\label{results}
\subsection{Bulk-Egde equality}
Our main result states that in the localization zone \eqref{loc region} of the Bulk operator $H_\omega$ and in presence of a confining edge, 
the Hall and edge conductances match and they are equal.
This result extends the main result of \cite{EGS} to the continuous setting and to purely random magnetic Schr\"odinger operators. 

\begin{theorem}\label{equality}
Let $I \subset (B_n,B_{n+1}) \cap \varSigma_{loc} $ be an interval for some $n\in\N$ such that $U_-,B_->\sup I$. 
Then for any switch function $g$ of $I$ and any $E\in \mathrm{supp}\ g'$,
the edge conductance \eqref{edge_cond} is well defined and we have
\begin{equation}
\sigma_{e,\omega}^\mathrm{reg}=\sigma_{\mathrm{Hall}}(B,\omega,E)  \ \ \mathrm{for} \ \proba-\mathrm{a.e} \ \omega . \end{equation}
\end{theorem}

\subsection{Strategy of the proof}\label{strategy}

Throughout the next sections, we fix $\omega\in\Omega$ and we let $I$ to be an interval such that $I\subset (B_n,B_{n+1})\cap \varSigma_{loc}$ for some $n\in\N$ given. 
Let $g$ be a switch function of $I$.
The core of the proof of Theorem~\ref{equality} is based on some intermediate steps that we shall outline below.
\bigskip

First, we compare the operators $g'(H_{\omega,a})\com{H_{\omega,a},\LL_2}$ and $\com{g(H_{\omega,a}),\LL_2} $ using
the Helffer-Sj\"{o}strand formulas but applied to the primitive function 
$$G(x):=\int_{x}^\infty g(s)\dd s.$$
We thus have
\begin{equation}\label{sj}
 g(H_{\omega,a})=-\frac{1}{2\pi}\int \overline{\partial}\tilde G(z) R_{\omega,a}^2 (z) \ \dd u \dd v 
\end{equation}
and
\begin{equation}\label{sj'}
 g'(H_{\omega,a})=\frac{1}{\pi}\int_{\R^2} \overline{\partial}\tilde G(z) R_{\omega,a}^3 (z) \ \dd u \dd v, 
\end{equation}
where $R_{\omega,a}(z)=(H_{\omega,a}-z)^{-1}$ and $z=u+iv$
and $\tilde G$ is a quasi-analytic extension of $G$ of order $k$ for $k=1,2,\dots$ \cite{D}.
Next, we use the second order resolvent identity 

\begin{equation}\label{R identity}
 \com{R_{\omega,a}^2(z),\LL_j}=- R_{\omega,a}^2(z)\com{H_{\omega,a},\LL_j}R_{\omega,a}(z)
-R_{\omega,a}(z)\com{H_{\omega,a},\LL_j}R_{\omega,a}^2(z)
\end{equation}
to write
\begin{align}
 \com{g(H_{\omega,a}),\LL_2}\LL_{1,a}^\omega(t)
&=\frac{1}{2\pi}\int_{\R^2} \overline{\partial}\tilde G(z) R_{\omega,a}^2(z) \com{H_{\omega,a},\LL_2} R_{\omega,a}(z) 
\LL_{1,a}^\omega(t)
\dd u \dd v\notag\\
&+\frac{1}{2\pi}\int_{\R^2} \overline{\partial}\tilde G(z) R_{\omega,a}(z) \com{H_{\omega,a},\LL_2} R_{\omega,a}^2(z)
\LL_{1,a}^\omega(t) \dd u \dd v.\label{g lambda}
\end{align}
We claim that the operators 
$$\com{g(H_{\omega,a}),\LL_2}\LL_{1,a}^\omega(t)\ \ \ \mathrm{and} \ \ \ g'(H_{\omega,a})\com{H_{\omega,a},\LL_2}\LL_{1,a}^\omega(t)$$
are both trace class according to Lemma~\ref{0 trace}. Together with \eqref{sj'} and the cyclicity of the trace, we have
\begin{align}
 \tr R_{\omega,a}^3(z)\com{H_{\omega,a},\LL_2}\LL_{1,a}^\omega(t)
&= \frac{1}{2}\tr R_{\omega,a}^2(z)\com{H_{\omega,a},\LL_2}\LL_{1,a}^\omega(t) R_{\omega,a}(z)\\
&+\frac{1}{2}\tr R_{\omega,a}(z)\com{H_{\omega,a},\LL_2}\LL_{1,a}^\omega(t)R_{\omega,a}^2(z).\label{g' lambda}
\end{align}
We thus compare \eqref{g lambda} and \eqref{g' lambda} and obtain an operator $\RR_{\omega,a}(t)$ that we call the remainder operator. We thus get
\begin{equation}\label{link}
\tr g'(H_{\omega,a})\com{H_{\omega,a},\LL_2}\LL_{1,a}^\omega(t)
=\tr \com{g(H_{\omega,a,}),\LL_2}\LL_{1,a}^\omega(t) +\tr \RR_{\omega,a}(t),
\end{equation}
with  
\begin{align}\label{Rest R}
\RR_{\omega,a}(t)
&=\frac{1}{2\pi}\int \overline{\partial} \tilde G(z) R_{\omega,a}^2(z) \com{H_{\omega,a},\LL_2} 
R_{\omega,a}(z) \com{H_{\omega,a},\LL_{1,a}^\omega(t)}R_{\omega,a}(z) \dd u \dd v \notag\\
&+ \frac{1}{2\pi}\int \overline{\partial}\tilde G(z) R_{\omega,a}(z) 
\com{H_{\omega,a},\LL_2}R_{\omega,a}^2(z) \com{ H_{\omega,a}, \LL_{1,a}^\omega(t)}R_{\omega,a}(z)  \dd u \dd v\notag\\
&+ \frac{1}{2\pi}\int \overline{\partial}\tilde G(z) R_{\omega,a}(z) 
\com{H_{\omega,a},\LL_2}R_{\omega,a}(z) \com{ H_{\omega,a}, \LL_{1,a}^\omega(t)}R_{\omega,a}^2(z)  \dd u \dd v.
\end{align}

We note that we have intentionally applied Helffer-Sj\"ostrand calculs to the primitive function $G$ in order to get sufficiently high power of the resolvent.

\bigskip

The key steps of the proof of Theorem~\ref{equality} are stated in forthcoming preliminary lemmas. 
The strategy consists in sending the wall to infinity by taking the limit $a\to+\infty$ in \eqref{link}. This leads to 
bulk quantities that we further average in time and analyze.
\bigskip

We start by showing that the key operators we deal with are trace class.
\begin{lemma}\label{0 trace}
Let $g$ be a switch function of an open interval $I$. Then the operators
\begin{enumerate}
 \item[$\bullet$] $ \com{g(H_{\omega,a}),\LL_2}\LL_1$
 \item[$\bullet$] $g'(H_{\omega,a})\com{H_{\omega,a},\LL_2}\LL_1$
 \item[$\bullet$] $ \com{g(H_{\omega,a}),\LL_2}\LL_{1,a}^\omega(t)$
 \item[$\bullet$] $g'(H_{\omega,a})\com{H_{\omega,a},\LL_2}\LL_{1,a}^\omega(t)$
\end{enumerate}
are trace classe for all $t\in\R$. Moreover, we have $\tr\com{g(H_{\omega,a}),\LL_2}\LL_1=0$.
\end{lemma}

The next lemma highlights the non-contribution of the remainder operator \eqref{Rest R}. 

\begin{lemma}\label{contribution R}
Let $I$ to be an interval such that $I\subset (B_n,B_{n+1})\cap \varSigma_{loc}$ for some $n\in\N$ given. Let $g$ a switch function of $I$. Then 
\begin{equation}\label{average R}
 \lim_{T\to\infty}\lim_{a\to\infty} \frac{1}{T}\int_0^T \tr\RR_{\omega,a}(t) \ \dd t = 0.
\end{equation}
\end{lemma}
\bigskip

We are thus left with the the first term of the r.h.s of \eqref{link}.
We rewrite the operator $\com{g(H_{\omega,a}),\LL_2}\LL_{1,a}^\omega(t)$ as
$\com{g(H_{\omega,a}),\LL_2}(\LL_{1,a}^\omega(t)-\LL_1)$,
since the operator $\com{g(H_{\omega,a}),\LL_2}\LL_1$ has zero trace by Lemma~\ref{0 trace}.  
This does not change the value of the trace but it provides a localization in space in the $x_1$-direction.

\begin{lemma}\label{limit in a}
Let $I$ to be an interval such that $I\subset (B_n,B_{n+1})\cap \varSigma_{loc}$ for some $n\in\N$ given. Let $g$ a switch function of $I$. Then
we have
\begin{equation}\label{cv in trace}
\lim_{a\to\infty}\tr\com{g(H_{\omega,a}),\LL_2}(\LL_{1,a}^\omega(t)-\LL_1) =\tr\com{g(H_\omega),\LL_2}(\LL_1^\omega(t)-\LL_1)
\end{equation}
for all $t\in\R$.
\end{lemma}
We can deal now with the resulting bulk expression and evaluate their contributions in time-average.

\begin{lemma}\label{average}
Let $I$ to be an interval such that $I\subset (B_n,B_{n+1})\cap \varSigma_{loc}$ for some $n\in\N$ given. Let $g$ a switch function of $I$. Then we have
 \begin{equation}\label{reg term to 0} 
\lim_{T\to\infty}\frac{1}{T}\int_0^T\tr\com{g(H_\omega),\LL_2}(\LL_1^\omega(t)-\LL_1) \
\dd t =\int g'(E) \tr \Pi_E \ \dd E,
\end{equation}
where 
\begin{equation}\label{pi_E}
 \Pi_E=\p\LL_2 \pp\LL_1 \p - \pp\LL_2 \p\LL_1 \pp \ \mathrm{and} \ \pp=1-\p.
\end{equation}

\end{lemma}
\bigskip

We point out how crucial it is to introduce $\LL_1$ in $ \com{g(H_{\omega,a}),\LL_2}\LL_{1,a}^\omega(t)$ for it gives a spatial 
localization in the $x_1$-direction by the difference $\LL_{1,a}^\omega(t)-\LL_1$. 
This makes the right operator in \eqref{cv in trace} trace class.
The proof of Lemma~\ref{average} actually shows that after having averaged in time, we only keep the term that comes from this added term $\LL_1$.

\bigskip

We now return to the Hall conductance \eqref{Hall conductane} which is directly connected to $\Pi_E$ defined in \eqref{pi_E} thanks to the following lemma. 

\begin{lemma}\label{dec Hall}
Let $I$ be an interval such that one has $I \subset (B_n,B_{n+1}) \cap  \Sigma_{loc}$ for some $n\in\N$. 
Then for any $E\in I$, we have
\begin{align}\label{dec Hall cond}
\sigma_{\mathrm{Hall}}(B,\omega,E)
&= i\tr( \p \LL_2 \pp \LL_1 \p - \pp \LL_2 \p \LL_1 \pp)\\
&=i\tr \Pi_E.\notag
\end{align}
\end{lemma}

Thanks to these preliminary lemmas and thanks to the assumption on $g$ and to the constancy of the Hall conductance in the localization region
\cite{GKS}, we thus deduce
\begin{equation}\label{conclusion}
 \sigma_{e,\omega}^\mathrm{reg}=\int g'(E) \ \sigma_\mathrm{Hall}(B,\omega,E)\ \dd E
=\sigma_\mathrm{Hall}(B,\omega,E),
\end{equation}
for any $E\in I \subset (B_n,B_{n+1}) \cap \Sigma_{loc}.$

\section{Proofs}\label{proofs}

In this section, we give the details of the proofs and intermediate steps. We start with the trace class property.
\subsection{Proof of Lemma~\ref{0 trace}}\label{tr operators}

We first deal with the operator $\com{g(H_{\omega,a}),\LL_2}\LL_1$ that we prove to be trace class with zero trace.
\subsubsection{Vanishing trace}\label{0 tr}
We first prove the vanishing trace for the pure magnetic model and we pursue with the electric one.

$\bullet$ \textit{Magnetic case}.
We proceed as in \cite{CG} and we split the operator 
$$\com{g(H_{\omega,a}^\m),\Lambda_2}\Lambda_1$$
in the $x_2$-direction so that for an arbitrary $R>0$, we write is as the sum of 
 \begin{equation}\label{I_R}
(\mathrm{I}_R)= \com{g(H_{\omega,a}^\m),\LL_2}\LL_1\ {\bf1}_{|x_2| \le R},
\end{equation}
and
\begin{equation}\label{II_R}
(\mathrm{II}_R)= \com{g(H_{\omega,a}^\m),\LL_2}\LL_1\ {\bf1}_{|x_2| > R},
\end{equation}
where ${\bf1_S}$ denotes the characteristic function of a subset $ S\subset\R^2$.
We first treat $\mathrm{I}_R$ in \eqref{I_R} that we decompose for $r>0$ as 
\begin{equation}\label{(I_R)}
 \com{g(H_{\omega,a}^\m),\LL_2}\ {\bf 1}_{|x_2| \le R} \ {\bf 1}_{-r_0-r-a\le x_1 \le 0} +
\com{g(H_{\omega,a}^\m),\LL_2}\ {\bf 1}_{|x_2| \le R} \ {\bf 1}_{x_1 \le -r_0-r-a}.
\end{equation}
We set $K={\bf 1}_{|x_2| \le R}\ {\bf 1}_{-r_0-r-a\le x_1 \le 0} $ appearing in the first term of the r.h.s of \eqref{(I_R)}. 
We notice that 
\begin{equation*}
 \com{g(H_{\omega,a}^\m),\LL_2}K=\com{g(H_{\omega,a}^\m)K,\LL_1}.
\end{equation*}
It is then sufficient to show that $g(H_{\omega,a}^\ee)K$ is a trace class operator and use the cyclicity of the trace to deduce 
immediately that 
\begin{equation}\notag
\tr\com{(g(H_{\omega,a}^\m),\LL_2} K=\tr\com{(g(H_{\omega,a}^\m)K,\LL_2}=0.
\end{equation}
To do this, it follows from the spectral theorem that
\begin{equation}
g(H_{\omega,a}^\m)=h(H_{\omega,a}^\m) \ \ \ \mathrm{with} \ \ h(s)=\chi_{\{s\ge1-\Theta\}} g(s), 
\end{equation}
where $\chi$ is a smooth characteristic function.
Notice that the function $h$ has compact support
($g$ verifies $\sup(\mathrm{supp}g')\ge1-\Theta$ otherwise $g(H_{\omega,a}^\m)=0$) and since $K$ has also compact support, 
we conclude that $h(H_{\omega,a}^\m) K=g(H_{\omega,a}^\m)K\in\T_1$ \cite[Theorem 4.1]{Si}. 
To prove a similar property for the remaining terms, we introduce a new operator. We let $\tilde \B$ to be a new magnetic field such that
\begin{equation*}
\tilde \B(x_1,x_2)\geq b_0> \sup \I \ \ \mathrm{for \ all} \ (x_1,x_2)\in\R^2, 
\end{equation*}
and it coincides with $\B$ for $x_1\le -r_0$, $r_0>0$. 
The difference magnetic field $\tilde\B-\B$ is then supported on $S_{r_0}:=\{x_1\ge -r_0\}\times \R$. It follows from \cite[Proposition 4.2]{DGR1}
that there exists is a magnetic potential $\tilde\A$ generating the magnetic field $\tilde\B-\B$ and vanishing on $S_{r_0}^c$.
Let us consider the auxiliary operator
\begin{equation}\label{H tilde}
\tilde H_{\omega,a}^\m:= (-i\nabla -\A_0-\A_a-\tilde \A_a -\A_\omega)^2,
 \end{equation}
where $\tilde\A_a(x)$ means $\tilde\A(x_1+a,x_2)$.
Since $\tilde H_{\omega,a}^\m-\tilde\B_a $ is a non-negative operator, it follows (see \cite{E}) that 
\begin{equation*}
\inf\sigma(\tilde H_{\omega,a}^\m)\geq \mathrm{Ess} \inf_{x_1\in\R} \tilde\B_a(x_1)\geq b_0.
\end{equation*}
As a consequence, one has $\sigma(\tilde H_{\omega,a}^\m)\cap \I =\emptyset$ and since $b_0>\sup I$ then $g(\tilde H_{\omega,a}^\m)=0$. 
We point out the creation of a forbidden zone where the electrons can not penetrate when we introduce such operators $\tilde H_{\omega,a}^\m$. 
We first consider the second term of the r.h.s of \eqref{(I_R)}, namely 
\begin{equation}\label{I_R 2}
\com{g(H_{\omega,a}^\m),\LL_2}\ {\bf 1}_{|x_2| \le R} \ {\bf 1}_{x_1 \le -r_0-r-a}
=\com{g(H_{\omega,a}^\m)-g(\tilde H_{\omega,a}^\m),\LL_2}\ {\bf 1}_{|x_2| \le R} \ {\bf 1}_{x_1 \le -r_0-r-a}.
\end{equation}
By the Helffer-Sj\"{o}strand formula, we have
\begin{equation}
 g(H_{\omega,a}^\m)-g(\tilde H_{\omega,a}^\m)=\frac{1}{2\pi}\int_{\R^2} \overline{\partial}\tilde g(u+iv)(R_{\omega,a,\m}-\tilde R_{\omega,a,\m}) \dd u \dd v.
\end{equation}
We thus have to analyze the operator $(R_{\omega,a,\m}-\tilde R_{\omega,a,\m})\ {\bf 1}_{|x_2| \le R} \ {\bf 1}_{x_1 \le -r_0-r-a}.$
We use commutators to check that
\begin{align*}
 &\com{R_{\omega,a,\m}-\tilde  R_{\omega,a,\m},\LL_2}
= \LL_2 \ R_{\omega,a,\m}\ \w_{\omega,a} \tilde R_{\omega,a,\m}-R_{\omega,a,\m} \ \w_{\omega,a} \tilde R_{\omega,a,\m} \ \LL_2 \\
& =R_{\omega,a,\m}\com{H_{\omega,a}^\m,\LL_2} R_{\omega,a,\m}\ \w_{\omega,a} \tilde R_{\omega,a,\m} + R_{\omega,a,\m}\ \LL_2 \w_{\omega,a} \tilde R_{\omega,a,\m}\\
& + R_{\omega,a,\m}\ \w_{\omega,a}\tilde R_{\omega,a,\m}\com{\tilde H_{\omega,a}^\m,\LL_2}\tilde R_{\omega,a,\m} - R_{\omega,a,\m}\ \w_{\omega,a}\LL_2 \tilde R_{\omega,a,\m}\\
&= R_{\omega,a,\m}\ \w_{\omega,a} \tilde R_{\omega,a,\m}\com{\tilde H_{\omega,a}^\m,\LL_2}\tilde R_{\omega,a,\m}
+ R_{\omega,a,\m}\com{H_{\omega,a}^\m,\LL_2} R_{\omega,a,\m}\ \w_{\omega,a} \tilde R_{\omega,a,\m}\\
&-R_{\omega,a,\m}\com{\w_{\omega,a},\LL_2}\tilde R_{\omega,a,\m},
\end{align*}
where the first-order operator $\w_{\omega,a}$ is given by
\begin{equation}\label{w_a magn}
 \w_{\omega,a}:=\tilde H_{\omega,a}-H_{\omega,a}=-2\tilde\A_a.(-i\nabla-\A_a^\mathrm{Iw}-\A_\omega)+i\dv \tilde\A_a+|\tilde\A_a|^2.
\end{equation}
\vspace{0.1cm}
\newline
We thus need to control the trace norms of  
\begin{equation}\label{com R W 0}
R_{\omega,a,\m}\com{H_{\omega,a}^\m,\LL_2} R_{\omega,a,\m}\ \w_{\omega,a} \tilde R_{\omega,a,\m}\ {\bf 1}_{|x_2| \le R} \ {\bf 1}_{x_1 \le -r_0-r-a},
\end{equation}
and
\begin{equation}\label{W R com 0}
R_{\omega,a,\m}\ \w_{\omega,a} \tilde R_{\omega,a,\m}\com{\tilde H_{\omega,a}^\m,\LL_2}\tilde R_{\omega,a,\m}\ {\bf 1}_{|x_2| \le R} \ {\bf 1}_{x_1 \le -r_0-r-a},
\end{equation}
and
\begin{equation}\label{com w}
 R_{\omega,a,\m}\com{\w_{\omega,a},\LL_2}\tilde R_{\omega,a,\m}\ {\bf 1}_{|x_2| \le R} \ {\bf 1}_{x_1 \le -r_0-r-a}.
\end{equation}
\vspace{0.3cm}
\newline
Now, having in mind that the commutator operators 
\begin{equation*}
 \com{H_{\omega,a}^\m,\LL_2}=-i(-i\nabla-\A_0-\A_a-\A_\omega). \nabla\LL_2- i\nabla\LL_2.(-i\nabla-\A_0-\A_a-\A_\omega),
\end{equation*}
and
\begin{equation*}
 \com{\tilde H_{\omega,a}^\m,\LL_2}=-i(-i\nabla-\A_0-\A_a-\tilde\A_a-\A_\omega). \nabla\LL_2- i\nabla\LL_2.(-i\nabla-\A_0-\A_a-\tilde\A_a-\A_\omega),
\end{equation*}
are localized on the support of $\nabla\LL_2$, we let $\chi_{|x_2|< 1}$ be a smooth characteristic function of $\R\times \{|x_2|<1\} $ so that we write 
\begin{equation}\label{com chi}
\com{H_{\omega,a}^\m,\LL_2}=\com{H_{\omega,a}^\m,\LL_2}\chi_{|x_2|< 1}\ \ \ \mbox{and}\ \ \ 
\com{\tilde H_{\omega,a}^\m,\LL_2}=\com{\tilde H_{\omega,a}^\m,\LL_2}\chi_{|x_2|< 1}.
 \end{equation}
We use unit cubes to decompose $\chi_{|x_2|< 1}$ as $$ \sum_{x=(x_1,0)\in\z^2} \chi_x,$$
where $(\chi_x)_{x\in\z^2}$ is a smooth decomposition of unity. 
We further let
\begin{equation}\label{w_a}
 \w_{\omega,a}= \sum_{\substack{y_1\in\z\cap[-a-r_0,\infty) \\ y_2\in\z}} \w_{\omega,a}\chi_y=
 \sum_{\substack{y_1\in\z\cap[-a-r_0,\infty) \\ y_2\in\z}} \chi_y\w_{\omega,a},
\end{equation}
and 
\begin{equation}\label{indic1}
 {\bf 1}_{|x_2| \le R} \ {\bf 1}_{x_1 \le -r_0-r-a} =\sum_{\substack{z_1\in\z\cap(\infty,-r_0-a-r]\\z_2\in\z\cap[-R,R]}} \chi_z.
\end{equation}
To treat \eqref{com R W 0}, we write

\begin{align}
&\eqref{com R W 0}
= R_{\omega,a,\m}\com{H_{\omega,a}^\m,\LL_2}\chi_{|x_2|\le1} R_{\omega,a,\m}\ \w_{\omega,a} \tilde R_{\omega,a,\m}\ {\bf 1}_{|x_2| \le R} \ {\bf 1}_{x_1 \le -r_0-r-a}\notag\\
& =\sum_{x,y,z\in\mathcal{S}_1} R_{\omega,a,\m}\com{H_{\omega,a}^\m,\LL_2}R_{\omega,a,\m}\ \chi_x \chi_y \w_{\omega,a} \tilde R_{\omega,a,\m}\chi_z\label{com R W 01}\\
&+ \sum_{x,y,z\in\mathcal{S}_1} R_{\omega,a,\m}\com{H_{\omega,a}^\m,\LL_2}R_{\omega,a,\m}\com{H_{\omega,a}^\m,\chi_x}R_{\omega,a,\m}\ 
\chi_y\w_{\omega,a} \tilde R_{\omega,a,\m} \chi_z,
\label{com R W 02}
\end{align}
with 
\begin{equation}\label{S_1}
\mathcal{S}_1:=\{\z\times\{0\} \} \times \{(\z\cap[-a-r_0,\infty) )\times\z \} \times \{(\z\cap(\infty,-r_0-a-r])\times(\z\cap[-R,R])\}.
\end{equation}
Notice that in \eqref{com R W 01}, we have $|y-x|\le2 $ and using Lemma~\ref{A.5} asserting that the operator 
$R_{\omega,a,\m}\com{H_{\omega,a}^\m,\LL_2}R_{\omega,a,\m} \chi_x $ 
is trace class independently of $x$, we obtain

\begin{align*}
\trnorm{\eqref{com R W 01}} 
&\le\sum_{x,y,z\in\mathcal{S}_1} \trnorm{ R_{\omega,a,\m}\com{H_{\omega,a}^\m,\LL_2}R_{\omega,a,\m}\ \chi_x \chi_y \w_{\omega,a} 
\tilde R_{\omega,a,\m}\chi_z}\\ 
&\le\sup_x\trnorm{R_{\omega,a,\m}\com{H_{\omega,a}^\m,\LL_2}R_{\omega,a,\m}\ \chi_x} \sum_{x,y,z\in\mathcal{S}_1}\norm{\chi_x \chi_y \w_{\omega,a} 
\tilde R_{\omega,a,\m}\chi_z}\\
&\le \frac{C}{\eta^2\tilde\eta}  \sum_{\substack{(x_1,x_2)\in{(\z\cap[-a-r_0-2,\infty))\times \{|x_2|< 1\}} 
\\ (z_1,z_2)\in(\z\cap(\infty,-r_0-a-r])\times(\z\cap[-R,R])}} \e^{-c\tilde\eta(|z_1-x_1| + |z_2-x_2|)}\\
&\le \frac{C_1}{\eta^2\tilde\eta} \ \e^{-c\tilde\eta r},
\end{align*}
where $\eta=\mathrm{dist}(z,\sigma(H_{\omega,a}^\m))$ and $ \tilde\eta=\mathrm{dist}(z,\sigma(\tilde H_{\omega,a}^\m))$. 
Since $r$ is arbitrary and for $R$ fixed, it follows that the trace vanishes.
Next, we estimate \eqref{com R W 02}. Let $\tilde\chi_x=1$ on the support of $\nabla\chi_x$.
Then, it suffices to use the decay of operator norms of 
$\chi_y\w_{\omega,a} \tilde R_{\omega,a,\m}\ \chi_z $ and of $\tilde\chi_x\com{H_{\omega,a},\chi_x}R_{\omega,a,\m}\chi_y$. 
We thus obtain
\begin{equation*}
\trnorm{\eqref{com R W 02}}
\le\sum_{x,y, z\in\mathcal{S}_1} \trnorm{R_{\omega,a,\m}\com{H_{\omega,a}^\m,\LL_2}R_{\omega,a,\m}\tilde\chi_x\com{H_{\omega,a}^\m,\chi_x}
R_{\omega,a,\m}\ \chi_y\w_{\omega,a}\tilde R_{\omega,a,\m}\ \chi_z}\\
\end{equation*}
\begin{align*}
&\le\sup_x\trnorm{R_{\omega,a,\m}\com{H_{\omega,a}^\m,\LL_2} R_{\omega,a,\m}\tilde\chi_x}\sum_{x,y,z\in\mathcal{S}_1}
||\tilde\chi_x\com{H_{\omega,a}^\m,\chi_x} R_{\omega,a,\m}\ \chi_y|| ||\chi_y\w_{\omega,a}\tilde R_{\omega,a,\m}\ \chi_z||\\ 
&\le\frac{c_1}{\eta^3\tilde\eta} \sum_{x,y,z\in\mathcal{S}_1} \e^{-c_2\eta(|y_1-x_1|+|y_2-x_2|)-\tilde c_2\tilde\eta(|z_1-y_1|+|z_2-y_2|) }\\
&\le\frac{\tilde c_1(a+r_0)}{\eta^3\tilde\eta} \e^{-c\tilde\eta r}.
\end{align*}
Taking $r\to\infty$, the trace of \eqref{com R W 02} vanishes and so does that of \eqref{com R W 0}.
A similar estimate holds for \eqref{W R com 0} so that we use \eqref{com chi}, \eqref{w_a} and \eqref{indic1} to write
\begin{align}
\eqref{W R com 0}
&=R_{\omega,a,\m}\ \w_{\omega,a} \tilde R_{\omega,a,\m}\com{\tilde H_{\omega,a}^\m,\LL_2}\tilde R_{\omega,a,\m}\ {\bf 1}_{|x_2| \le R} \ {\bf 1}_{x_1 \le -r_0-r-a}\notag\\
&=\sum_{x,y,z\in\mathcal{S}_1} R_{\omega,a,\m}\w_{\omega,a}\tilde R_{\omega,a,\m}\chi_y\chi_x\com{\tilde H_{\omega,a},\LL_2}\tilde R_{\omega,a,\m}\chi_z\label{W R com 01}\\
&+ \sum_{x,y,z\in\mathcal{S}_1}R_{\omega,a,\m}\w_{\omega,a}\tilde R_{\omega,a,\m}\com{\tilde H_{\omega,a}^\m,\chi_y}\tilde R_{\omega,a,\m}\chi_x\com{\tilde H_{\omega,a}^\m,\LL_2}
\tilde R_{\omega,a,\m}\chi_z,\label{W R com 02}
\end{align}
where $\mathcal{S}_1$ is defined in \eqref{S_1}. The trace class property holds for \eqref{W R com 01} from $ R_{\omega,a}\w_{\omega,a}\tilde R_{\omega,a}\chi_y$
so that
\begin{align*}
\trnorm{\eqref{W R com 01}}
&\le\sum_{x,y,z\in\mathcal{S}_1}\trnorm{R_{\omega,a,\m}\w_{\omega,a}\tilde R_{\omega,a,\m}\chi_y\chi_x\com{\tilde H_{\omega,a}^\m,\LL_2}\tilde R_{\omega,a,\m}\chi_z}\\
&\le\sup_y\trnorm{R_{\omega,a,\m}\w_{\omega,a}\tilde R_{\omega,a,\m}\chi_y}
\sum_{x,y,z\in\mathcal{S}_1}\norm{\chi_y\chi_x\com{\tilde H_{\omega,a}^\m,\LL_2}\tilde R_{\omega,a,\m}\chi_z}\\
&\le\frac{c_2}{\tilde\eta^2\eta} \sum_{\substack{(x_1,x_2)\in{(\z\cap[-a-r_0-2,\infty))\times \{|x_2|< 1\}} 
\\ (z_1,z_2)\in(\z\cap(\infty,-r_0-a-r])\times(\z\cap[-R,R])}} \e^{-c\tilde\eta(|z_1-x_1| + |z_2-x_2|)}\\
&\le \frac{\tilde c_2}{\eta^2\tilde\eta} \ \e^{-c\tilde\eta r}.
\end{align*} 
We have the analog procedure for \eqref{W R com 02} since we let $\tilde\chi_y=1$ on $\mathrm{supp}\ \nabla\chi_y$ and thus
\begin{equation*}
 \trnorm{\eqref{W R com 02}}
 \le\sum_{x,y,z\in\mathcal{S}_1}\trnorm{R_{\omega,a,\m}\w_{\omega,a}\tilde R_{\omega,a,\m}\tilde\chi_y\com{\tilde H_{\omega,a}^\m,\chi_y}
 \tilde R_{\omega,a,\m}\chi_x\com{\tilde H_{\omega,a}^\m,\LL_2}\tilde R_{\omega,a,\m}\chi_z}\\
\end{equation*}

\begin{align*}
 &\le \sup_y\trnorm{R_{\omega,a,\m}\w_{\omega,a}\tilde R_{\omega,a,\m}\tilde\chi_y}\sum_{x,y,z\in\mathcal{S}_1}
 \norm{\chi_y\com{\tilde H_{\omega,a}^\m,\chi_y}\tilde R_{\omega,a,\m}\chi_x}\norm{\chi_x\com{\tilde H_{\omega,a}^\m,\LL_2}\tilde R_{\omega,a,\m}\chi_z}\\
&\le\frac{c_3}{\tilde\eta^3\eta}\sum_{x,y,z\in\mathcal{S}_1}\e^{-c\tilde\eta(|x_1-y_1|+|x_2-y_2|+|z_1-x_1|+|z_2-x_2|)}\\
&\le\frac{\tilde c_3 (a+r_0)}{\tilde\eta^3\eta} \e^{-c\tilde\eta r}.
\end{align*}
Since $r$ is arbitrary, it follows that the trace of \eqref{W R com 01} and \eqref{W R com 02} vanish.
Next, we estimate the trace norm of \eqref{com w}.
Since $\com{\w_{\omega,a},\LL_2}=2i\tilde\A_a^{(2)}.\nabla\LL_2$, we have
\begin{align}
&\eqref{com w}
= R_{\omega,a,\m}\com{\w_{\omega,a},\LL_2}\tilde R_{\omega,a,\m} {\bf 1}_{|x_2| \le R} \ {\bf 1}_{x_1 \le -r_0-r-a}\notag\\
&=\sum_{x,y,z\in\mathcal{\tilde S}_1}R_{\omega,a,\m}\com{\w_{\omega,a},\LL_2}\chi_y\tilde R_{\omega,a,\m}\chi_z\notag\\
&=\sum_{x,y,z\in\mathcal{\tilde S}_1}\left(R_{\omega,a,\m}\com{\w_{\omega,a},\LL_2}\tilde R_{\omega,a,\m}\chi_y\chi_z
+ R_{\omega,a,\m}\com{\w_{\omega,a},\LL_2}\tilde R_{\omega,a,\m}\tilde\chi_y\com{\tilde H_{\omega,a}^\m,\chi_y}\tilde R_{\omega,a,\m}\chi_z\right),\label{com w 0}
\end{align}
where 
\begin{equation}\label{tilde S_1}
 \mathcal{\tilde S}_1:=\{(\z\cap[-a-r_0,\infty) )\times\{0\} \} \times \{(\z\cap(\infty,-r_0-a-r])\times(\z\cap[-R,R])\}.
\end{equation}
Since $y$ and $z$ lie in disjoint supports, we notice that the l.h.s of \eqref{com w 0} zero. It remains thus to deal with 
$R_{\omega,a,\m}\com{\w_{\omega,a},\LL_2}\tilde R_{\omega,a,\m}\chi_y\com{\tilde H_{\omega,a}^\m,\chi_y}\tilde R_{\omega,a,\m}\chi_z $ so that
\begin{align*}
\trnorm{\eqref{com w}}
&\le\sup_y\trnorm{R_{\omega,a,\m}\com{\w_{\omega,a},\LL_2}\tilde R_{\omega,a,\m}\tilde\chi_y}\sum_{y,z\in\mathcal{\tilde S}_1} 
\norm{\tilde\chi_y\com{\tilde H_{\omega,a}^\m,\chi_y}\tilde R_{\omega,a,\m}\chi_z}\\
&\le\frac{c_4}{\eta\tilde\eta^2}\sum_{y,z\in\mathcal{\tilde S}_1}  \e^{-c\tilde\eta(|z_1-y_1|+|z_2-y_2|)}\\
&\le\frac{\tilde c_4}{\eta\tilde\eta^2} \e^{-c\tilde\eta r},
\end{align*}
which goes to 0 for $r$ arbitrarily chosen.
We deal now with the term $(\mathrm{II}_R)$ in \eqref{II_R} that we treat exactly in the same way as \eqref{I_R 2}. Indeed, following the previous steps, we have to check the trace norm of
\begin{equation}\label{com R w 0}
R_{\omega,a,\m}\com{H_{\omega,a}^\m,\LL_2} R_{\omega,a,\m} \w_{\omega,a} \tilde R_{\omega,a,\m}\ \LL_1 {\bf 1}_{|x_2|>R},
\end{equation}
and
\begin{equation}\label{w R com 0}
R_{\omega,a,\m} \w_{\omega,a} \tilde R_{\omega,a,\m}\com{\tilde H_{\omega,a}^\m,\LL_2}\tilde R_{\omega,a,\m}\ \LL_1 {\bf 1}_{|x_2|>R},
\end{equation}
and
\begin{equation}\label{com W}
R_{\omega,a,\m}\com{\w_{\omega,a},\LL_2}\tilde R_{\omega,a,\m}.
\end{equation}
We write
\begin{equation}\label{indic2} 
\LL_1 {\bf 1}_{|x_2|>R}=\sum_{(z_1,z_2)\in(\z^-\times(\z\cap[-R,R]^c)} \chi_z,
 \end{equation}
and we start with \eqref{com R w 0} that we express as
\begin{align}
&\eqref{com R w 0}
= R_{\omega,a,\m}\com{H_{\omega,a}^\m,\LL_2}\chi_{|x_2|\le1} R_{\omega,a,\m}\w_{\omega,a} \tilde R_{\omega,a,\m}\ \LL_1 {\bf 1}_{|x_2|>R}\notag\\
&=\sum_{x,y,z\in\mathcal{S}_2}R_{\omega,a,\m}\com{H_{\omega,a}^\m,\LL_2}\chi_x R_{\omega,a,\m}\chi_y\w_{\omega,a} \tilde R_{\omega,a,\m}\chi_z\notag\\
& =\sum_{x,y,z\in\mathcal{S}_2} R_{\omega,a,\m}\com{H_{\omega,a}^\m,\LL_2}R_{\omega,a,\m}\ \chi_x \chi_y\w_{\omega,a} \tilde R_{\omega,a,\m}\chi_z\label{com R w 01}\\
&+ \sum_{x,y,z\in\mathcal{S}_2} R_{\omega,a,\m}\com{H_{\omega,a}^\m,\LL_2}R_{\omega,a,\m}\com{H_{\omega,a}^\m,\chi_x}R_{\omega,a,\m}\ \chi_y\w_{\omega,a}
\tilde R_{\omega,a,\m} \chi_z,\label{com R w 02}
\end{align}
similarly to \eqref{com R W 0}, where 
\begin{equation}\label{S_2}
 \mathcal{S}_2 :=\{\z\times \{0\} \} \times \{(\z\cap[-a-r_0,\infty) )\times\z \} \times \{\z^-\times(\z\cap[-R,R]^c)\}.
\end{equation}
To estimate the trace norm of \eqref{com R w 01} and \eqref{com R w 02}, we follow \eqref{com R W 01} and \eqref{com R W 02} and we get
\begin{align*}
\trnorm{\eqref{com R w 01}} 
&\le\sum_{x,y,z\in\mathcal{S}_2} \trnorm{ R_{\omega,a,\m}\com{H_{\omega,a}^\m,\LL_2}R_{\omega,a,\m}\ \chi_x \chi_y 
\w_{\omega,a}\tilde R_{\omega,a,\m}\chi_z}\notag\\ 
&\le \frac{C}{\tilde\eta^2\tilde\eta}  \sum_{\substack{(x_1,x_2)\in{(\z\cap[-a-r_0-2,\infty))\times \{|x_2|\le 1\}} 
\\ (z_1,z_2)\in\z^-\times(\z\cap[-R,R]^c)}} \e^{-c\tilde\eta(|z_1-x_1| + |z_2-x_2|)}\\
&\le \frac{C_1}{\eta^2\tilde\eta} \ \e^{c\tilde\eta(a+r_0)} \e^{-c\tilde\eta R},
\end{align*}
and
\begin{align*}
\trnorm{\eqref{com R w 02}}
&\le\sum_{x,y, z\in\mathcal{S}_2} \trnorm{R_{\omega,a,\m}\com{H_{\omega,a}^\m,\LL_2}R_{\omega,a,\m}\com{H_{\omega,a}^\m,\chi_x}
\tilde\chi_x R_{\omega,a,\m}\ \chi_y\tilde R_{\omega,a,\m}\ \chi_z}\\
&\le\frac{c_1}{\eta^2\tilde\eta} \sum_{x,y,z\in\mathcal{S}_2} \e^{-c_2\eta(|y_1-x_1|+|y_2-x_2|)-c_2\tilde\eta(|z_1-y_1|+|z_2-y_2|) }\\
&\le\frac{\tilde c_1 (a+r_0)}{\eta^2\tilde\eta} \  \e^{-c\tilde\eta R}.
\end{align*}
Since $R$ is arbitrary, we conclude that the traces of \eqref{com R w 01} and \eqref{com R w 02} vanish.
Similarly to \eqref{W R com 0}, we have
\begin{align}
&\eqref{w R com 0}
=R_{\omega,a,\m}\ \w_{\omega,a} \tilde R_{\omega,a,\m}\com{\tilde H_{\omega,a}^\m,\LL_2}\tilde R_{\omega,a,\m}\ \LL_1 {\bf 1}_{|x_2|>R}\notag\\
&=\sum_{x,y,z\in\mathcal{S}_2} R_{\omega,a,\m}\w_{\omega,a}\tilde R_{\omega,a,\m}\chi_y\chi_x\com{\tilde H_{\omega,a},\LL_2}\tilde R_{\omega,a,\m}\chi_z\label{w R com 01}\\
&+ \sum_{x,y,z\in\mathcal{S}_2}R_{\omega,a,\m}\w_{\omega,a}\tilde R_{\omega,a,\m}\com{\tilde H_{\omega,a}^\m,\chi_y}\tilde R_{\omega,a,\m}\chi_x\com{\tilde H_{\omega,a}^\m,\LL_2}
\tilde R_{\omega,a,\m}\chi_z,\label{w R com 02}
\end{align}
where the set $\mathcal{S}_2$ is defined in \eqref{S_2}. The trace norms of \eqref{w R com 01} and \eqref{w R com 02} are estimated similarly to that of 
\eqref{W R com 01} and \eqref{W R com 02}, so that one has
\begin{align*}
\trnorm{\eqref{w R com 01}}
&\le\sum_{x,y,z\in\mathcal{S}_2}\trnorm{R_{\omega,a,\m}\w_{\omega,a}\tilde R_{\omega,a,\m}\chi_y\chi_x\com{\tilde H_{\omega,a}^\m,\LL_2}\tilde R_{\omega,a,\m}\chi_z}\\
&\le\frac{c_2}{\tilde\eta^2\eta} \sum_{\substack{(x_1,x_2)\in{(\z\cap[-a-r_0-2,\infty))\times \{|x_2|< 1\}} 
\\ (z_1,z_2)\in\z^{-}\times(\z\cap[-R,R]^c)}} \e^{-c\tilde\eta(|z_1-x_1| + |z_2-x_2|)}\\
&\le \frac{\tilde c_2}{\eta^2\tilde\eta} \ \e^{c\tilde\eta(a+r_0)} \e^{-c\tilde\eta R},
\end{align*} 
and 
\begin{align*}
 \trnorm{\eqref{w R com 02}}
& \le\sum_{x,y,z\in\mathcal{S}_1}\trnorm{R_{\omega,a,\m}\w_{\omega,a}\tilde R_{\omega,a,\m}\tilde\chi_y\com{\tilde H_{\omega,a}^\m,\chi_y}
 \tilde R_{\omega,a,\m}\chi_x\com{\tilde H_{\omega,a}^\m,\LL_2}\tilde R_{\omega,a,\m}\chi_z}\\
&\le\frac{c_3}{\tilde\eta^3\eta}\sum_{x,y,z\in\mathcal{S}_1}\e^{-c\tilde\eta(|x_1-y_1|+|x_2-y_2|+|z_1-x_1|+|z_2-x_2|)}\\
&\le\frac{\tilde c_3 (a+r_0)}{\tilde\eta^3\eta} \e^{-c\tilde\eta R},
\end{align*}
with $\tilde\chi_y=1$ on $\mathrm{supp}(\nabla\chi_y)$. Next, we finish with \eqref{com W} which is similar to \eqref{com w} in the sense that
\begin{align*}
 \trnorm{\eqref{com W}}
 &\le\sup\trnorm{R_{\omega,a,\m}\com{\w_{\omega,a},\LL_2}\tilde R_{\omega,a,\m}\tilde\chi_y}
 \sum_{y,z\in\mathcal{\tilde S}_2}\norm{\tilde\chi_y\com{\tilde H_{\omega,a}^\m,\chi_y}\tilde R_{\omega,a,\m}\chi_z}\\
&\le\frac{c_4}{\eta\tilde\eta^2}\sum_{y,z\in\mathcal{\tilde S}_1}  \e^{-c\tilde\eta(|z_1-y_1|+|z_2-y_2|)}\\
&\le\frac{\tilde c_4}{\eta\tilde\eta^2} \e^{-c\tilde\eta R},
 \end{align*}
where 
\begin{equation}\label{tilde S_2}
 \mathcal{\tilde S}_2:=\{(\z\cap[-a-r_0,\infty) )\times\{0\} \} \times \{\z^{-}\times(\z\cap[-R,R]^c)\}.
\end{equation}
Since $R$ is arbitrarily chosen, we deduce that the traces of \eqref{w R com 0} and \eqref{com W} are equal to zero. 

\bigskip

$\bullet$ \textit{Electric case}.
For the reader's convenience, we sketch the main steps of the previous proof for the electric model.
We split the operator $\com{g(H_{\omega,a}^\ee,\LL_2}\LL_1$ in the $x_2$-direction
such that  
\begin{equation}\label{dec elec}
\com{g(H_{\omega,a}^\ee,\LL_2}\LL_1=\com{g(H_{\omega,a}^\ee),\LL_2}\LL_1 {\bf 1}_{\{|x_2|\le R\}} + \com{g(H_{\omega,a}^\ee),\LL_2}\LL_1 {\bf 1}_{\{|x_2|> R\}},
\end{equation}
For an arbitrary $R>0$. In order to extract a compact part, we decompose the first r.h.s of \eqref{dec elec} in the $x_1$-direction for $r>0$ arbitrary 
and we write it as 
\begin{equation}\label{dec elec 2}
\com{g(H_{\omega,a}^\ee),\LL_2} {\bf 1}_{\{|x_2|\le R\}} {\bf 1}_{\{-r_0-a-r\le x_1\le 0\}}
+\com{g(H_{\omega,a}^\ee),\LL_2} {\bf 1}_{\{|x_2|\le R\}} {\bf 1}_{\{x_1\le -r_0-a-r\}}.
\end{equation}
The trace of the l.h.s of \eqref{dec elec 2} is zero, following the magnetic case. 
The assumptions on the electric potential $U$ yields that there exists $r_0>0$ such that
\begin{equation} U_a(x)\geq c_0, \ \forall x_1 <-a-r_0, \end{equation}
where $c_0$ is choosen so that $c_0 > \sup \I$. 
The auxiliary operator that we consider is
\begin{equation}\label{H tilde elec}
\tilde{H}^\ee_{\omega,a}:= H_{\omega,a}^\ee + \ c_0 \ {\bf 1}_{x_1 \geq -r_0-a}.
\end{equation}
In particular, $g(\tilde{H}^\ee_{\omega,a})=0$ since its spectrum is disjoint from  $\I$. 
Otherwise, to treat the second terms in r.h.s of \eqref{dec elec 2} and \eqref{dec elec}, we take advantage of the auxiliary operator 
$\tilde H_{\omega,a,\m}$ defined in \eqref{H tilde elec}, as we did for \eqref{I_R 2} and \eqref{I_R}, except that the first operator $\w_{\omega,a}$ is 
replaced by the operator $W_a$ given by
\begin{equation}
 W_a:=\tilde H_{\omega,a}^\ee-H_{\omega,a}^\ee=c_0 \ {\bf 1}_{x_1 \geq -r_0-a}.
\end{equation}
We start by  $ \com{g(H_{\omega,a}^\ee),\LL_2} {\bf 1}_{\{|x_2|\le R\}} {\bf 1}_{\{x_1\le -r_0-a-r\}}$ that leads to check the trace norm of 
\begin{equation}\label{com R w}
 R_{\omega,a,\ee}\com{H_{\omega,a}^\ee,\LL_2}R_{\omega,a,\ee} W_a\tilde R_{\omega,a,\ee}{\bf 1}_{\{|x_2|\le R\}} {\bf 1}_{\{x_1\le -r_0-a-r\}}
\end{equation}
and 
\begin{equation}\label{com w R}
 R_{\omega,a,\ee} W_a\tilde R_{\omega,a,\ee}\com{H_{\omega,a}^\ee,\LL_2}R_{\omega,a,\ee}{\bf 1}_{\{|x_2|\le R\}} {\bf 1}_{\{x_1\le -r_0-a-r\}}. 
\end{equation}
Once more, we use smooth decomposition of unity and we write 
$$ W_a= \sum_{\substack{y_1\in\z\cap(\infty,-r_0-a]\\ y_2\in\z}} \chi_z.$$
The term in \eqref{com R w} is treated in the same way as \eqref{com R W 0} where we have to estimate the trace norm of 
\begin{equation}
 \sum_{x,y,z\in\mathcal{S}_1} R_{\omega,a,\ee}\com{H_{\omega,a}^\ee,\LL_2}R_{\omega,a,\ee}\ \chi_x \chi_y \tilde R_{\omega,a,\ee}\chi_z,
\end{equation}
and 
\begin{equation}
\sum_{x,y,z\in\mathcal{S}_1} R_{\omega,a,\ee}\com{H_{\omega,a}^\ee,\LL_2}R_{\omega,a,\ee}\com{H_{\omega,a}^\ee,\chi_x}R_{\omega,a,\ee}\ \chi_y
\tilde R_{\omega,a,\ee} \chi_z.
\end{equation}
This follows from \eqref{com R W 01} and \eqref{com R W 02} where we use Lemma~\ref{resolvent} to obtain a decay of the kernel of $\w_{\omega,a}\tilde R_{\omega,a,\m}$
instead of the Combes-Thomas estimate. In particular, \eqref{com w R} follows the same procedure as \eqref{W R com 0}.
\vspace{0.1cm}
\newline
We now turn to the remaining term $ \com{g(H_{\omega,a}^\m)-g(\tilde H_{\omega,a}^\m),\LL_2}\LL_1 {\bf 1}_{\{|x_2|> R\}}$ which is similar to \eqref{II_R}. 
This gives analogous terms to \eqref{com R w 0} and \eqref{w R com 0} where once more, the decay of the kernel of the resolvent is replaced by that of 
$\w_{\omega,a} R_{\omega,a,\m}$ thanks to Lemma~\ref{resolvent}. After all, we conclude that the trace of $\com{g(H_{\omega,a}^\m),\LL_2}\LL_1$ vanishes.

\subsubsection{Trace class property}
In this section, we deal with the trace class property of the opertors mentionned in Lemma~\ref{0 trace}.
\vspace{0.3cm}
\newline
\textit{The operator $ g'(H_{\omega,a})\com{H_{\omega,a},\LL_2,}\LL_1$}. 
The trace class property of this operator follows from the previous section where we have considered auxiliary operators $\tilde H_{\omega,a}$ to take avantage of the wall.
Since we have $g'(\tilde H_{\omega,a})=0$, we should analyze the operator 
$$ (g'(H_{\omega,a})-g'(\tilde H_{\omega,a}))\com{H_{\omega,a},\LL_2,}\LL_1$$
via the formula \eqref{sj}.
\bigskip

$\bullet$ \textit{Magnetic case}.
After computation and recalling that $\w_{\omega,a}=\tilde H_{\omega,a}^\m-H_{\omega,a}^\m$, we obtain six terms
\begin{align*}
 (\tilde R_{\omega,a,\m}^3&-R_{\omega,a,\m}^3)\com{H_{\omega,a}^\m,\LL_2}\LL_1\\
&=\left(\tilde R_{\omega,a,\m}^2 \w_{\omega,a} R_{\omega,a,\m}\tilde R_{\omega,a,\m} 
+\tilde R_{\omega,a,\m} \w_{\omega,a} R_{\omega,a,\m}^2\tilde R_{\omega,a,\m}\right)\com{H_{\omega,a}^\m,\LL_2}\LL_1\\
&+\left(\tilde R_{\omega,a,\m} \w_{\omega,a} R_{\omega,a,\m}\tilde R_{\omega,a,\m}^2
+R_{\omega,a,\m}^2 \tilde R_{\omega,a,\m} \w_{\omega,a} R_{\omega,a,\m}\right)\com{H_{\omega,a}^\m,\LL_2}\LL_1\\ 
&+\left( R_{\omega,a,\m}\tilde R_{\omega,a,\m}^2 \w_{\omega,a} R_{\omega,a,\m}+
R_{\omega,a,\m}\tilde R_{\omega,a,\m} \w_{\omega,a} R_{\omega,a,\m}^2\right)\com{H_{\omega,a}^\m,\LL_2}\LL_1.
\end{align*}
We shall treat one term and the others holds in quite similar way. For instance, we deal with 
$ R_{\omega,a,\m}\tilde R_{\omega,a,\m} \w_{\omega,a} R_{\omega,a,\m}^2\com{H_{\omega,a}^\m,\LL_2}\LL_1$ that we write as the sum of
\begin{equation}\label{wr2}
R_{\omega,a,\m} \tilde R_{\omega,a,\m} \chi_y \w_{\omega,a}R_{\omega,a,\m}^2\com{H_{\omega,a}^\m,\LL_2}\chi_x
\end{equation}
over
$$\mathcal{D}_1:=\{(x_1,x_2)\in\z^-\times(\z\cap[-1,1]), (y_1,y_2)\in(\z\cap(\infty,-a-r_0])\times\z, u\in\z^2\},$$
so that
\begin{equation}
\style\trnorm{\eqref{wr2}}
\le\trnorm{R_{\omega,a,\m} \tilde R_{\omega,a,\m}\chi_y}\norm{\chi_y \w_{\omega,a} R_{\omega,a,\m}\chi_u}\norm{\chi_u R_{\omega,a,\m}
\com{H_{\omega,a}^\m,\LL_2}\chi_x}.\label{est}
\end{equation}
Since $R_{\omega,a,\m} \tilde R_{\omega,a,\m}\chi_y$ is trace class independently of $y$ and having in mind that $\w_{\omega,a}$ is a first order operator, 
we use Lemma~\ref{resolvent} 
to upper bound \eqref{est} by 
$$c_1 |\im z|^{-4} \ \e^{-c_2|\im z|(|u-y|+|x-u|)}.$$

\bigskip

$\bullet$ \textit{Electric case}.
Once more, the same arguments work for the electric case subject to change $\w_{\omega,a}$ into $W_a$. 
If we consider the term $ R_{\omega,a,\ee}\tilde R_{\omega,a,\ee}^2 W_a R_{\omega,a,\ee}\com{H_{\omega,a}^\ee,\LL_2}\LL_1$, 
we have to estimate the trace norm of the sum of 

\begin{equation}
R_{\omega,a,\ee}\tilde R_{\omega,a,\ee}^2 \chi_y R_{\omega,a,\ee}\com{H_{\omega,a}^\ee,\LL_2}\chi_x
\end{equation}
over 
$$\mathcal{D}_2:=\{(x_1,x_2)\in\z^-\times\{0\}, (y_1,y_2)\in(\z\cap(\infty,-a-r_0])\times\z\}.$$
Thus the trace class property holds from $\tilde R_{\omega,a,\ee}^2 \chi_y$ while the summabilty of the sum comes out from the decay of 
$\chi_y R_{\omega,a,\ee}\com{H_{\omega,a}^\ee,\LL_2}\chi_x$ thanks to the Combes-Thomas estimate \cite{CT}. In the sense that 
\begin{equation}
 \trnorm{R_{\omega,a,\ee}\tilde R_{\omega,a,\ee}^2 W_a R_{\omega,a,\ee}\com{H_{\omega,a}^\ee,\LL_2}\LL_1}\le C(a+r_0) |\im z|^{-4}.
\end{equation}

\bigskip

Now we establish the trace class property of operators depending on the time regularization $\LL_{1,a}^\omega(t)$.
\vspace{0.3cm}
\newline
\textit{The operator $\com{g(H_{\omega,a}),\LL_2,}\LL_{1,a}^\omega(t)$}.
In next analysis, we do not need to specify the case we deal with since the proof works for both electric and magnetic models. 
We substract $\com{g(H_{\omega,a}),\LL_2}\LL_1$ which has zero trace by the previous analysis in section~\ref{0 tr}. Moreover, combining
\begin{equation}\label{intgr}
 \LL_{1,a}^\omega(t) -\LL_1 = i\int_0^t \e^{isH_{\omega,a}}\com{H_{\omega,a},\LL_1}\e^{-isH_{\omega,a}} \dd s,
\end{equation}
that we insert in \eqref{sj} and the resolvent identity \eqref{R identity},
we are left with the analysis of the trace norm of
\begin{equation}\label{R² com}
R_{\omega,a}^2\com{H_{\omega,a},\LL_2}R_{\omega,a} \ \e^{isH_{\omega,a}}\com{H_{\omega,a},\LL_1}\e^{-isH_{\omega,a}},
\end{equation}
and 
\begin{equation}\label{com R²}
R_{\omega,a}\com{H_{\omega,a},\LL_2}R_{\omega,a}^2 \ \e^{isH_{\omega,a}}\com{H_{\omega,a},\LL_1}\e^{-isH_{\omega,a}}.
\end{equation}
These operators are localized in space in both directions $x_1$ and $x_2$ in the sense that each $\com{H_{\omega,a},\LL_j}$ is 
localized on the support of $\LL_j'$ because
\begin{align}
\com{H_{\omega,a}^\m,\LL_j}
&=-i(-i\nabla-A_0-A_a-A_\omega). \nabla\LL_j- i\nabla\LL_j.(-i\nabla-A_0-A_a-A_\omega)\notag
\end{align}
and
\begin{equation*}
\com{H_{\omega,a}^\ee,\LL_j}=-i(-i\nabla-A_0). \nabla\LL_j- i\nabla\LL_j.(-i\nabla-A_0).
\end{equation*}
To estimate the trace norm of \eqref{R² com}, we decompose it with smooth characteristic functions and we rewrite
\begin{equation}
 \eqref{R² com} 
=\sum_{\substack{(x_1,x_2)\in\z\times\{0\}\\ (y_1,y_2)\in\{0\}\times\z}}
R_{\omega,a}^2 \com{H_{\omega,a},\LL_2}\chi_x R_{\omega,a}\ \e^{isH_{\omega,a}} 
\com{H_{\omega,a},\LL_1} \chi_y \ \e^{-isH_{\omega,a}}.\label{sum R² com}
\end{equation}
Since the operator $R_{\omega,a}^2\com{H_{\omega,a},\LL_2}\chi_x$ is trace class with 
\begin{equation*}\trnorm{R_{\omega,a}^2\com{H_{\omega,a},\LL_2}\chi_x}\le\frac{C}{|\im z|^2},\end{equation*}
and the operator norm of $\chi_x R_{\omega,a}\ \e^{isH_{\omega,a}}\com{H_{\omega,a},\LL_1}\chi_y $ is upper bounded by 

\begin{equation*}
 \e^{c_1 s}|\im z|^{-1} \ \e^{-c_2|\im z|(|x_1-y_1|+|x_2-y_2|)},
\end{equation*}
which follows from Lemma~\ref{resolvent}, we obtain the sum \eqref{sum R² com} is finite and thus \eqref{R² com} is trace class.
Moreover, there exist two constants $c_1$ and $c_2$ such that
\begin{equation}
 \trnorm{\eqref{R² com}} \le c_3 |\im z|^{-3} \ \e^{c_1 s}.\label{time bnd 1}
\end{equation}
We turn to \eqref{com R²} that we expand in the following way

\begin{equation*}
 \eqref{com R²}=\sum_{\substack{u_1,y_2\in\z,\ x\in\z^2\\ u_2=0,y_1=0}}
R_{\omega,a} \com{H_{\omega,a},\LL_2}\chi_u R_{\omega,a} \chi_x \ \e^{isH_{\omega,a}} R_{\omega,a}
\com{H_{\omega,a},\LL_1} \chi_y \ \e^{-isH_{\omega,a}}.
\end{equation*}
In order to extract the decay in $x_1$ and $y_2$, we use commutators to push $\chi_u $ to the left through the resolvent $R_{\omega,a}$.
Let $\tilde\chi_u$ be a smooth function such that $\tilde\chi_u=1$ on $\mathrm{supp}\nabla\chi_u$. Then we have
\begin{align*}
&\trnorm{ \eqref{com R²}} 
\le\sum_{\substack{u_1,y_2\in\z,\ x\in\z^2\\ u_2=0,y_1=0}}
\trnorm{R_{\omega,a} \com{H_{\omega,a},\LL_2}\chi_u R_{\omega,a} \chi_x \ \e^{isH_{\omega,a}} R_{\omega,a}
\com{H_{\omega,a},\LL_1} \chi_y \ \e^{-isH_{\omega,a}}}\\
&\le\sum_{\substack{u_1,y_2\in\z,\ x\in\z^2\\ u_2=0,y_1=0}}
\trnorm{ R_{\omega,a} \com{H_{\omega,a},\LL_2} R_{\omega,a} \chi_u \chi_x \ \e^{isH_{\omega,a}} R_{\omega,a}
\com{H_{\omega,a},\LL_1} \chi_y \ \e^{-isH_{\omega,a}}}\\
&+\sum_{\substack{u_1,y_2\in\z,\ x\in\z^2\\ u_2=0,y_1=0}}
\trnorm{R_{\omega,a} \com{H_{\omega,a},\LL_2} R_{\omega,a} \tilde\chi_u\com{H_{\omega,a},\chi_u} R_{\omega,a}\chi_x \ \e^{isH_{\omega,a}} R_{\omega,a}
\com{H_{\omega,a},\LL_1} \chi_y \ \e^{-isH_{\omega,a}}}\\
&\le  C_1  |\im z|^{-3} \ \e^{c_1 s} \sum_{\substack{x_1,y_2\in\z\\ y_1=0\\x_2\in\z\cap[-2,2]}} 
\e^{-\tilde c_1 |\im z|(|x_1-y_1|+|x_2-y_2|)}\\
&+ C_2  |\im z|^{-4} \ \e^{c_2 s} \sum_{\substack{u_1,y_2\in\z,x\in\z^2\\ u_2=0,y_1=0}} 
\e^{-\tilde c_2 |\im z|(|x_1-u_1|+|x_2-u_2|+|x_1-y_1|+|x_2-y_2|)}\\
&\le \tilde C_1  |\im z|^{-3} \ \e^{c_1 s} + C_2  |\im z|^{-4} \ \e^{c_2 s} \sum_{\substack{u_1,y_2\in\z\\u_2=0, y_2=0}} 
\e^{- \tilde c_2 |\im z|(|u_1-y_1|+|u_2-y_2|)},
\end{align*}
where we have combined Combes-Thomas estimate \cite{CT} and Lemma~\ref{resolvent} together with Lemma~\ref{A.5}.
Hence, the summability follows and the operator \eqref{com R²} is finally trace class with 
\begin{equation}
\trnorm{\eqref{com R²}}\le C_3  \ \e^{\tilde c_3 s} (|\im z|^{-3}+|\im z|^{-4}). \label{time bnd 2}
\end{equation}
\vspace{0.3cm}
\newline
\textit{The operator $ g'(H_{\omega,a})\com{H_{\omega,a},\LL_2}\LL_{1,a}^\omega(t)$}.
From \eqref{sj'}, it follows that
\begin{align}
g'(H_{\omega,a})&\com{H_{\omega,a},\LL_2}\LL_{1,a}^\omega(t)
= g'(H_{\omega,a})\com{H_{\omega,a},\LL_2}\LL_1\notag\\
&+\frac{i}{\pi}\int_{\R^2}\int_0^t \overline{\partial}\tilde G(z) R_{\omega,a}^3(z)\com{H_{\omega,a},\Lambda_2}\e^{isH_{\omega,a}}\com{H_{\omega,a},\LL_1}
\e^{-isH_{\omega,a}}\dd s \ \dd u \ \dd v. \label{integ}
\end{align}
By the previous result on the operator $g'(H_{\omega,a})\com{H_{\omega,a},\LL_2}\LL_{1,a}^\omega(t)$, it suffices to treat \eqref{integ} and to estimate the trace norm operator of
\begin{align}
R_{\omega,a}^3\com{H_{\omega,a},\LL_2}\e^{isH_{\omega,a}}\com{H_{\omega,a},\LL_1}
\e^{-isH_{\omega,a}},
\end{align}
that we write as
\begin{align}
& R_{\omega,a}^2\com{H_{\omega,a},\LL_2}R_{\omega,a}\e^{isH_{\omega,a}}\com{H_{\omega,a},\LL_1}\e^{-isH_{\omega,a}}\label{term 1}\\
&-R_{\omega,a}^3\com{H_{\omega,a},\com{H_{\omega,a},\LL_2}}R_{\omega,a}\ \e^{isH_{\omega,a}}\com{H_{\omega,a},\LL_1}\e^{-isH_{\omega,a}}.\label{term 2}
\end{align}
We thus have
\begin{align*}
\trnorm{\eqref{term 1}}
&\le\sum_{\substack{x_1,y_2\in\z,\ x\in\z^2\\ x_2=0,y_1=0}}
\trnorm{R_{\omega,a}^2\com{H_{\omega,a},\LL_2}\chi_x}\norm{\chi_x R_{\omega,a}\e^{isH_{\omega,a}}\com{H_{\omega,a},\LL_1}\chi_y},
\end{align*}
and 
\begin{align*}
 \trnorm{\eqref{term 2}}
&\le\sum_{\substack{x_1,y_2\in\z,\ x\in\z^2\\ x_2=0,y_1=0}}
 \trnorm{R_{\omega,a}^3\com{H_{\omega,a},\com{H_{\omega,a},\LL_2}}\chi_x}\norm{\chi_x R_{\omega,a}\ \e^{isH_{\omega,a}}\com{H_{\omega,a},\LL_1}\chi_y}.
\end{align*}
The trace norms above are upper bounded by a constant $c$ uniformly in $x$ and the operator norms operators are bounded by 
$$\e^{c_1 s} |\im z|^{-1} \e^{-c_2 |\im z|(|x_1-y_1|+|x_2-y_2|)}.$$
Then the trace class property holds.

\bigskip

Notice that although the operator $\com{g(H_{\omega,a}),\LL_2}\LL_{1,a}^\omega(t)$ is still trace class, there is no reason anymore for its trace to vanishes since
$\LL_2$ does not commute with $\LL_{1,a}^\omega(t)$ as it is the case with $\LL_1$.  

\subsection{Contributions of the Bulk quantities}

We start by proving the zero contribution of the remainder term \eqref{Rest R}.

\subsubsection{Proof of Lemma~\ref{contribution R}}\label{limits}
 For convenience we set
\begin{equation}\label{r_1}
 r_{\omega,a}^{(1)}(t)=R_{\omega,a}^2(z) \com{H_{\omega,a},\LL_2} 
R_{\omega,a}(z) \com{H_{\omega,a},\LL_{1,a}^\omega(t)}R_{\omega,a}(z), 
\end{equation}

\begin{equation}\label{r_2}
 r_{\omega,a}^{(2)}(t)=R_{\omega,a}(z) \com{H_{\omega,a},\LL_2} 
R_{\omega,a}^2(z) \com{H_{\omega,a},\LL_{1,a}^\omega(t)}R_{\omega,a}(z),
\end{equation}
and
\begin{equation}\label{r_3}
 r_{\omega,a}^{(3)}(t)=R_{\omega,a}(z) \com{H_{\omega,a},\LL_2} 
R_{\omega,a}(z) \com{H_{\omega,a},\LL_{1,a}^\omega(t)}R_{\omega,a}^2(z),
\end{equation}
that appear in \eqref{Rest R}.
We first treat \eqref{r_2} and prove the convergence to the corresponds bulk quantity. Rewrite $r_{\omega,a}^{(2)}t)$ as
\begin{equation}\label{rr_2}
\com{R_{\omega,a},\LL_2} (H_\omega+\Theta) \scal{x_2}^{2\nu} (\scal{x_2}^{-2\nu} 
(H_\omega+\Theta)^{-2} \scal{x_1}^{-2\nu}) \scal{x_1}^{2\nu} (H_\omega+\Theta)\com{R_{\omega,a},\LL_{1,a}^\omega(t)}.
\end{equation}
We notice that the operators 
\begin{equation*}
 \com{R_{\omega,a},\LL_2} (H_\omega+\Theta) \scal{x_2}^{2\nu} \ \ \mbox{and} \ \
\scal{x_1}^{2\nu} (H_\omega+\Theta)\com{R_{\omega,a},\LL_{1,a}^\omega(t)}
\end{equation*}
are uniformly bounded in $a$.
As the middle operator $(\scal{x_2}^{-2\nu} (H_\omega+\Theta)^{-2} \scal{x_1}^{-2\nu})$ is trace class (see \cite{BoGKS}), it follows from Lemma~\ref{stron lim} 
and Proposition~\ref{si} that it suffices to prove the strong convergence of the left and right operators in \eqref{rr_2} in $\mathcal{C}_c^\infty(\R^2)$.
We use the identity $ R_{\omega,a}^\ee -R_\omega^\ee=-R_{\omega,a}^\ee U_a \ R_\omega^\ee$ to write
\begin{equation}\label{com(R_a-R)}
 \com{R_{\omega,a}^\ee-R_{\omega}^\ee,\LL_2} =
\LL_2 R_{\omega,a}^\ee U_a \ R_\omega^\ee  
-R_{\omega,a}^\ee U_a \ R_\omega^\ee \ \LL_2. 
\end{equation}
Similarly, we have
\begin{equation*}
 R_{\omega,a}^\m -R_\omega^\m=-R_{\omega,a}^\m \Gamma_{\omega,a} \ R_\omega^\m,
\end{equation*}
for the magnetic model and thus
\begin{equation}\label{com(R_a-R)}
 \com{R_{\omega,a}^\m-R_{\omega}^\m,\LL_2} =
\LL_2 R_{\omega,a}^\m \Gamma_{\omega,a} \ R_\omega^\m 
-R_{\omega,a}^\m \Gamma_a \ R_\omega^\m \ \LL_2 . 
\end{equation}
Let
$\varphi\in\cc_c^\infty(\R^2)$ be such that $\mathrm{supp}\varphi \subset D_{r_1,r_2}$  where $D_{r_1,r_2}=[-r_1,r_1]\times[-r_2,r_2] $ for $r_1<a$ and $r_2 >0$. 
In particular, $\mathrm{supp}(\LL_2(H_\omega+\Theta)\scal{x_2}^{2\nu}\varphi)\subset D_{r_1,r_2}$
and we have
\begin{equation*}
 \norm{R_{\omega,a}^\m \ \Gamma_{\omega,a} \ R_\omega^\m \ \LL_2 (H_\omega+\Theta) \scal{x_2}^{2\nu}\varphi}
\le \frac{C}{|\im z|^2} \ \e^{-c |a-r_1|} \norm{\LL_2 (H_\omega+\Theta) \scal{x_2}^{2\nu}\varphi}
\end{equation*}
and
\begin{equation*}
 \norm{\LL_2 \ R_{\omega,a}^\m \ \Gamma_{\omega,a} \ R_\omega^\m (H_\omega+\Theta) \scal{x_2}^{2\nu}\varphi}
\le \frac{C}{|\im z|^2} \ \e^{-c |a-r_1|} \norm{(H_\omega+\Theta) \scal{x_2}^{2\nu}\varphi}, 
\end{equation*}
which converge to 0 as $ a\to+\infty$. The electric case holds in a the same way.
\vspace{0.3cm}
\newline
Next we carry on the convergence of the right side of the operator in \eqref{rr_2} and we write
\begin{equation}\label{com(R_a-R,t)}
\com{R_{\omega,a},\LL_{1,a}^\omega(t)}-\com{R_\omega,\LL_{1}^\omega(t)}=\com{R_{\omega,a}-R_\omega,\LL_{1,a}^\omega(t)}+
\com{R_\omega,\LL_{1,a}^\omega(t)-\LL_1^\omega(t)}.
\end{equation}
We point out that the first term of the r.h.s of \eqref{com(R_a-R,t)} is treated in the same spirit as \eqref{com(R_a-R)} without time-dependence.
In fact, one has
\begin{equation}\label{Ra R}
 \com{R_{\omega,a}-R_\omega,\LL_{1,a}^\omega(t)}=\com{R_{\omega,a}-R_\omega,\LL_{1,a}^\omega(t)-\LL_1}+\com{R_{\omega,a}-R_\omega,\LL_1},
\end{equation}
and the second term of the r.h.s of \eqref{Ra R} looks like \eqref{com(R_a-R)} where we have $\LL_1$ instead of $\LL_2$. For the first term of \eqref{Ra R}, 
we take advantage of localisation in $x_1$ that the difference $\LL_{1,a}^\omega(t)-\LL_1$ gives us (see \eqref{intgr}) and the result holds similarly.
\vspace{0.3cm}
\newline
We come back to the second term in the r.h.s of \eqref{com(R_a-R,t)}, namely
$$\scal{x_1}^{2\nu}(H_\omega+\Theta)\com{R_{\omega},\LL_{1,a}^\omega(t)-\LL_1^\omega(t)},$$ 
that requires more works. We combine the commutator calculation and the first order resolvent identity to obtain
\begin{equation*}
 \com{R_{\omega},\LL_1^\omega(t)}=-R_\omega \e^{itH_\omega}\com{H_\omega,\LL_1} \e^{-itH_\omega} R_\omega,
\end{equation*}
and
\begin{align*}
 \LL_{1,a}^\omega(t) R_\omega
&=  \LL_{1,a}^\omega(t) R_{\omega,a}(1+\Gamma_{\omega,a} R_\omega)\\
&= R_{\omega,a} \ \LL_{1,a}^\omega(t)(1+\Gamma_{\omega,a} R_\omega)+ \e^{-itH_{\omega,a}} \com{R_{\omega,a},\LL_1} \e^{-itH_{\omega,a}}(1+\Gamma_{\omega,a} R_\omega).
\end{align*}
Hence, one has
\begin{align*}
 \com{R_{\omega},\LL_{1,a}^\omega(t)}
&=  R_\omega \LL_{1,a}^\omega(t) - R_{\omega,a} \ \LL_{1,a}^\omega(t)(1+\Gamma_{\omega,a} R_\omega)\\
&- \e^{-itH_{\omega,a}} \com{R_{\omega,a},\LL_1} \e^{-itH_{\omega,a}}
(1+\Gamma_{\omega,a} R_\omega),
\end{align*}
that we plug into $\com{R_\omega,\LL_{1,a}^\omega(t)-\LL_{1}^\omega(t)}$ to get
\begin{align*}
 \com{R_\omega,\LL_{1,a}^\omega(t)-\LL_{1}^\omega(t)}
&=(R_\omega-R_{\omega,a}) \LL_{1,a}^\omega(t)- R_{\omega,a}\LL_{1,a}^\omega(t)\Gamma_{\omega,a} R_\omega\\
&- e^{-itH_{\omega,a}} \com{R_{\omega,a},\LL_1} \e^{-itH_{\omega,a}}(1+\Gamma_{\omega,a} R_\omega)
- \e^{itH_\omega}\com{R_\omega,\LL_1} \e^{-itH_\omega}.
\end{align*}
Hence, as $\LL_{1,a}^\omega(t)\to\LL_{1}^\omega(t)$ and $R_\omega\Gamma_{\omega,a}\to 0$ strongly, by Lemma~\ref{strong cv R}, 
the strong convergence to $0$ as $a\to\infty$ follows.
\vspace{0.3cm}
\newline
Now, we deal with \eqref{r_1} and push one resolvent from the left through the commutator $\com{H_{\omega,a},\LL_1}$, so that
\begin{align}\label{rr_1}
 r_{\omega,a}^{(1)}(t)
&=R_{\omega,a}\com{R_{\omega,a},\LL_2}R_{\omega,a}^2 \com{H_{\omega,a},\LL_{1,a}^\omega(t)}R_{\omega,a}\\
&-R_{\omega,a}^2 \com{H_{\omega,a},\com{H_{\omega,a},\LL_2}} R_{\omega,a}^2 \com{H_{\omega,a},\LL_{1,a}^\omega(t)}R_{\omega,a}.\label{rr_12}
\end{align}
The first term \eqref{rr_1} fit exactly to \eqref{r_1}. Procceding as in \eqref{rr_2} we get
\begin{align*}
\eqref{rr_12}
=-R_{\omega,a}&\com{R_{\omega,a},\com{H_{\omega,a},\LL_2}}(H_\omega+\Theta) \scal{x_2}^{2\nu} (\scal{x_2}^{-2\nu} 
(H_\omega+\Theta)^{-2} \scal{x_1}^{-2\nu})\\
 &\scal{x_1}^{2\nu} (H_\omega+\Theta) \com{R_{\omega,a},\LL_{1,a}^\omega(t)}.
\end{align*}
Once more, the middle operator $\scal{x_2}^{-2\nu} (H_\omega+\Theta)^{-2} \scal{x_1}^{-2\nu}$
is trace class \cite{BoGKS}.
By Lemma~\ref{stron lim} and Proposition~\ref{si} together with Lemma~\ref{strong cv R}
and the fact that the right operator above $\scal{x_1}^{2\nu} (H_\omega+\Theta) \com{R_{\omega,a},\LL_{1,a}^\omega(t)}$
 is previousely treated in \eqref{com(R_a-R)}, we only need to prove the strong convergence of the operator
$$\com{R_{\omega,a},\com{H_{\omega,a},\LL_2}}(H_\omega+\Theta) \scal{x_2}^{2\nu},$$
which is uniformly bounded in $a$.
We compute the difference
\begin{equation}\label{rr 12}
\com{R_{\omega,a},\com{H_{\omega,a},\LL_2}}- \com{R_{\omega},\com{H_{\omega},\LL_2}}
=\com{R_{\omega,a}-R_\omega, \com{H_{\omega,a},\LL_2}} +\com{R_\omega,\com{\Gamma_{\omega,a},\LL_2}},
\end{equation}
and we let $\varphi\in\mathcal{C}_c^\infty(\R^2)$ with support $D_{r_1,r_2}=[-r_1,r_1]\times[-r_2,r_2]$ for $r_1<a$ and $r_2>0$. Then the supports of 
$\com{H_{\omega,a},\LL_2}(H_\omega+\Theta)\scal{x_2}^{2\nu}\varphi$ and $(H_\omega+\Theta)\scal{x_2}^{2\nu}\varphi$ are both contained in $D_{r_1,r_2}$.
Thus we estimate the operator norm of
\begin{equation}\label{rr 12 bis}
\norm{\com{R_{\omega,a}-R_\omega, \com{H_{\omega,a},\LL_2}}(H_\omega+\Theta) \scal{x_2}^{2\nu}\varphi}
\end{equation}
so that
\begin{align*}
\eqref{rr 12 bis}&\le \norm{ R_\omega\Gamma_{\omega,a}R_{\omega,a}\com{H_{\omega,a},\LL_2}(H_\omega+\Theta) \scal{x_2}^{2\nu}\varphi}\\
&+ \norm{\com{H_{\omega,a},\LL_2}R_{\omega,a}\Gamma_{\omega,a}R_\omega(H_\omega+\Theta) \scal{x_2}^{2\nu}\varphi}\\
&\le\norm{R_\omega} \norm{\Gamma_{\omega,a}R_{\omega,a}\com{H_{\omega,a},\LL_2}(H_\omega+\Theta) \scal{x_2}^{2\nu}\varphi}\\
&+\norm{\com{H_{\omega,a},\LL_2}R_{\omega,a}}\norm{\Gamma_{\omega,a}R_\omega(H_\omega+\Theta) \scal{x_2}^{2\nu}\varphi}
\\
&\le C \left( |\im z|^{-2}\ \e^{-\tilde c_1|\im z||a-r_1|}+ |\im z|^{-3/2}\ \e^{-\tilde c_2|\im z||a-r_1|}\right)\norm{(H_\omega+\Theta)\scal{x_2}^{2\nu}\varphi},
\end{align*}
which converges to $0$ as $a\to\infty$. 
Consider now the remaining term $ \com{R_\omega,\com{\Gamma_{\omega,a},\LL_2}}$ of the r.h.s of \eqref{rr 12}. We have
\begin{align*}
 \com{R_\omega,\com{\Gamma_{\omega,a},\LL_2}}
=R_\omega\com{\Gamma_{\omega,a},\LL_2}-\Gamma_{\omega,a}\LL_2R_\omega+\LL_2\Gamma_{\omega,a}R_\omega,
\end{align*}
and control its operator norm in the following way
\begin{equation}\label{3}
 \norm{\Gamma_{\omega,a}\LL_2R_\omega (H_\omega+\Theta)\scal{x_2}^{2\nu}}\le c_3 |\im z|^{-1} \e^{-\tilde c_3|\im z||a-r_1|}\norm{(H_\omega+\Theta)\scal{x_2}^{2\nu}\varphi},
\end{equation}
and
\begin{equation}\label{4}
 \norm{\LL_2\Gamma_{\omega,a}R_\omega(H_\omega+\Theta)\scal{x_2}^{2\nu}}\le c_4 |\im z|^{-1} \e^{-\tilde c_4|\im z||a-r_1|}\norm{(H_\omega+\Theta)\scal{x_2}^{2\nu}\varphi},
\end{equation}
while $R_\omega\com{\Gamma_{\omega,a},\LL_2}(H_\omega+\Theta)\scal{x_2}^{2\nu}\varphi=0$ since $r_1<a$.
We thus conclude that \eqref{3} and \eqref{4} converge to 0 as $a\to\infty$. 
\vspace{0.4cm}
\newline
In a similar way, we can establish the strong convergences in $a$ of \eqref{r_3} to the bulk corresponding operators
such that $ r_{\omega,a}^{(3)}(t)\to r_{\omega}^{(3)}(t)$ where 
we denote by $r_\omega^{(1)}(t), r_\omega^{(2)}(t)$ and $r_\omega^{(3)}(t)$ the analougous remainders. 
\vspace{0.4cm}
\newline
Next, we estimate the time average of $r_\omega^{(j)}(t) $ in the trace norm for $j=1,2,3$.
In the first step, we introduce smooth characteristic functions $\chi_{\{|x_j|\le R\}}$ and $\chi_{\{|x_j|> R\}}$ inside $r_\omega^{(1)}(t)$ where $R=T^{1/2}$
and $j=1,2$. 
\newline
We rewrite $r_\omega^{(1)}(t)$ as the sum  
\begin{equation}\label{intg r_1}
 {R_{\omega}^2 \com{H_{\omega},\LL_2} (\chi_{\{|x_1|\le R\}}+\chi_{\{|x_1|> R\}})R_{\omega}\com{H_{\omega},\LL_{1}^\omega(t)} R_{\omega}}.
\end{equation}
We consider the time average of the r.h.s of \eqref{intg r_1} whose trace norm is estimated as
\begin{align*}
\frac{1}{T}&\trnorm{R_{\omega}^2 \com{H_{\omega},\LL_2} \chi_{\{|x_1|\le R\}}R_{\omega}\left(\e^{iTH_\omega}\LL_1 \e^{-iTH_\omega}-\LL_1\right) R_{\omega}} \\
&\le\frac{1}{T}\trnorm{R_{\omega}^2 \com{H_{\omega},\LL_2} \chi_{\{|x_1|\le R\}}}\norm{R_{\omega}\left(\e^{iTH_\omega}\LL_1 \e^{-iTH_\omega}-\LL_1\right) R_{\omega}}\\
&\le \frac{C R}{T} |\im z|^{-4},
\end{align*}
which goes to $0$ as $T\to\infty$ and where we have used the fact that operator $R_{\omega}^2 \com{H_{\omega},\LL_2}\chi_{\{|x_1|\le R\}}$ belongs to $\T_1$ together
with
\begin{align}
\frac{1}{T}\int_0^T \com{H_{\omega},\LL_{1}^\omega(t)}\ \dd t
&=\frac{1}{T}\int_0^T \e^{itH_\omega}\com{H_\omega,\LL_1}\e^{-itH_\omega} \dd t\notag\\
&=\frac{-i}{T}(\e^{iTH_\omega}\LL_1 \e^{-iTH_\omega} - \LL_1).\notag
\end{align}
Concerning the second term of the r.h.s of \eqref{intg r_1}, we have
\begin{align*}
& \trnorm{R_{\omega}^2 \com{H_{\omega},\LL_2}\left(\frac{1}{T}\int_0^T\chi_{\{|x_1|> R\}} R_{\omega} \e^{itH_\omega}\com{H_\omega,\LL_1}\e^{-itH_\omega} R_\omega\ \dd t\right)}\\
&\le\sum_{x,y\in\mathcal{N}_1} \trnorm{\frac{1}{T}\int_0^T R_{\omega}^2 \com{H_{\omega},\LL_2}\chi_x  R_{\omega} \e^{itH_\omega}\chi_y
\com{H_\omega,\LL_1}\e^{-itH_\omega} R_\omega\ \dd t}\\
&\le \sum_{x,y\in\mathcal{N}_1} \trnorm{R_{\omega}^2 \com{H_{\omega},\LL_2}\chi_x} \left(\frac{1}{T}\int_0^T\norm{\chi_x R_{\omega} \e^{itH_\omega}\chi_y}\ \dd t\right)
\norm{\com{H_\omega,\LL_1}\e^{-itH_\omega} R_\omega}\\
&\le \tilde C |\im z|^{-4} \ T^5 \e^{-c_1|\im z|R},
\end{align*}
where 
\begin{equation}\label{N_1}
\mathcal{N}_1 =\{(\z\cap[-R,R]^c \times \{0\}) \times (\{0\}\times\z)\}. 
\end{equation}
Here, we have used the decay of the kernel $\chi_x R_{\omega} \e^{itH_\omega}\chi_y $.
Since $R=T^{\frac12}$, the trace thus vanishes as $T\to\infty$.
\vspace{0.3cm}
\newline
The result $r_\omega^{(2)}(t) $ and $r_\omega^{(3)}(t)$ follows in quite similar way.
For the reader's convenience, we nevertheless reproduce the details for $r_\omega^{(2)}(t)$.

\begin{align}
 r_\omega^{(2)}(t)
&=R_\omega\com{H_\omega,\LL_2}R_\omega\chi_{\{|x_1|\le R\}}R_\omega\com{H_\omega,\LL_1^\omega(t)}R_\omega\label{r_2 1}\\
&+R_\omega\com{H_\omega,\LL_2}R_\omega\chi_{\{|x_1|< R\}}R_\omega\com{H_\omega,\LL_1^\omega(t)}R_\omega\label{r_2 2}
\end{align}

\begin{align*}
 \frac{1}{T}\int_0^T\trnorm{\eqref{r_2 1}}\ \dd t
 &\le\frac{1}{T}\trnorm{R_\omega\com{H_\omega,\LL_2}R_\omega\chi_{\{|x_1|\le R\}}} \norm{R_\omega(\e^{iTH_\omega}\LL_1\e^{iTH_\omega}-\LL_1)R_\omega}\\
 &\le \frac{cR}{T}|\im z|^{-4}.
\end{align*}

\begin{align*}
& \frac{1}{T}\int_0^T\trnorm{\eqref{r_2 2}}\ \dd t
\le\sum_{x,y\in\mathcal{N}_2}\frac{1}{T}\trnorm{R_\omega\com{H_\omega,\LL_2}R_\omega\chi_x R_\omega\e^{itH_\omega}\chi_y\com{H_\omega,\LL_1}\e^{-itH_\omega}R_\omega} \dd t\\
&\le\sup_x\trnorm{R_\omega\com{H_\omega,\LL_2}R_\omega\chi_x}\sum_{x,y\in\mathcal{N}}\left(\frac{1}{T}\norm{\chi_x R_\omega\e^{itH_\omega}\chi_y} \dd t\right)
\norm{\com{H_\omega,\LL_1}\e^{-itH_\omega}R_\omega}\\
&\le c_1 T^5 |\im z|^{-4} \ \e^{-c_2|\im z| R},
\end{align*}
where 
\begin{equation}\label{N_2}
 \mathcal{N}_2=\{(\z\cap[-R,R]^c \times \z) \times (\{0\}\times\z)\}.
\end{equation}
To conclude, we take the function $\tilde G$ of order 5 so that the limit \eqref{average R} follows.

\subsubsection{Proof of Lemma~\ref{limit in a}}
It follows from the section~\ref{limits} that the operator $\com{g(H_{\omega,a}),\Lambda_2}(\Lambda_{1,a}^\omega(t)-\Lambda_1) $ is trace class. 
Concerning the convergence in trace to $\com{g(H_{\omega}),\Lambda_2}(\Lambda_{1}^\omega(t)-\Lambda_1) $, we adopt the same techniques used for the remainder operator 
$\mathcal{R}_{\omega,a}(t)$ in section~\ref{limits}. 
We use \eqref{sj} and we notice that is enough to analyze the operators
\begin{equation}\label{lim 1}
 R_{\omega,a}\com{H_{\omega,a},\LL_2}R_{\omega,a}^2(\LL_{1,a}^\omega(t)-\LL_1)
\end{equation}
and
\begin{equation}\label{lim 2}
 R_{\omega,a}^2\com{H_{\omega,a},\LL_2}R_{\omega,a}(\LL_{1,a}^\omega(t)-\LL_1).
\end{equation}
\bigskip
Once again, we introduce the operator $(H_\omega+\Theta)^2$ inside \eqref{lim 1} and \eqref{lim 2}. We write
\begin{align}
\eqref{lim 1}
=- \com{R_{\omega,a},\LL_2}&(H_\omega+\Theta)\scal{x_2}^{2\nu}(\scal{x_2}^{-2\nu}(H_\omega+\Theta)^2\scal{x_1}^{-2\nu})\notag\\
&\scal{x_1}^{2\nu}(H_\omega+\Theta)R_{\omega,a}(\LL_{1,a}^\omega(t)-\LL_1).\label{lim 1 1}
\end{align}
Since the operator $(\scal{x_2}^{-2\nu}(H_\omega+\Theta)^2\scal{x_1}^{-2\nu}) $ is trace class \cite{BoGKS},
it suffices thanks to Lemma~\ref{stron lim}, to prove the strong convergence of
\begin{equation*}
\com{R_{\omega,a},\LL_2}(H_\omega+\Theta)\scal{x_2}^{2\nu} \ \ \mbox{and} \ \ \scal{x_1}^{2\nu}
(H_\omega+\Theta)R_{\omega,a}(\LL_{1,a}^\omega(t)-\LL_1), 
\end{equation*}
in $\mathcal{C}_c^\infty(\R^2)$ as they are bounded uniformly in $a$.
We notice that the operator $\com{R_{\omega,a},\LL_2}(H_\omega+\Theta)\scal{x_2}^{2\nu}$ has already been treated in \eqref{r_1}. 
We are now left with $\scal{x_1}^{2\nu}(H_\omega+\Theta)R_{\omega,a}(\LL_{1,a}^\omega(t)-\LL_1)$ that we rewrite as 
$$\scal{x_1}^{2\nu} (H_\omega+\Theta) (R_{\omega,a}-R_\omega) (\LL_{1,a}^\omega(t)-\LL_1) + \scal{x_1}^{2\nu} (H_\omega+\Theta) R_\omega (\LL_{1,a}^\omega(t)-\LL_1).$$
Since $\LL_{1,a}^\omega(t)\to \LL_{1}^\omega(t) $ strongly, the second term converges to zero. To see that
$\scal{x_1}^{2\nu}(H_\omega+\Theta)(R_{\omega,a}-R_\omega)$ converges strongly to 0, we use
$R_{\omega,a}-R_\omega=-R_{\omega}\Gamma_{\omega,a}R_{\omega,a} $ and we let $\varphi\in\mathcal{C}_c^\infty(\R^2)$ compactly supported in $D_{r_1,r_2}$ as in section~\ref{limits}
with $r_1<a$. Then
\begin{align*}
 \norm{\scal{x_1}^{2\nu}(H_\omega+\Theta)R_{\omega}\Gamma_{\omega,a}R_{\omega,a}\varphi}
&\le \scal{r_1}^{2\nu}\norm{(H_\omega+\Theta)R_{\omega}} \norm{\Gamma_{\omega,a}R_{\omega,a}\varphi}\\
&\le \frac{C_{r_1}}{|\im z|^{1/2}} \e^{-c|\im z||a-r_1|} \norm{\varphi},
\end{align*}
which converges to 0 as $a\to\infty$.
\vspace{0.3cm}
\newline
Next we turn to \eqref{lim 2} whose analysis will be similar to that of \eqref{r_1}. We commute $R_{\omega,a}$ and $\com{H_{\omega,a},\LL_2}$ to write
\begin{align}
\eqref{r_1}
&=R_{\omega,a}^2\com{H_{\omega,a},\LL_2}R_{\omega,a}(\LL_{1,a}^\omega(t)-\LL_1)\notag\\
&=R_{\omega,a}\com{H_{\omega,a},\LL_2}R_{\omega,a}^2(\LL_{1,a}^\omega(t)-\LL_1)\label{01}\\
&-R_{\omega,a}^2\com{H_{\omega,a},\com{H_{\omega,a},\LL_2}}R_{\omega,a}^2(\LL_{1,a}^\omega(t)-\LL_1).\label{02}
\end{align}
Since the first term \eqref{01} fit excatly to \eqref{lim 1}, we only need to check \eqref{02}. We have
\begin{align*}
 \eqref{02}
=&R_{\omega,a}\com{R_{\omega,a},\com{H_{\omega,a},\LL_2}}(H_\omega+\Theta)\scal{x_2}^{2\nu}(\scal{x_2}^{-2\nu}(H_\omega+\Theta)^{-2} \scal{x_1}^{-2\nu})\\
&\scal{x_1}^{2\nu}(H_\omega+\Theta)R_{\omega,a}(\LL_{1,a}^\omega(t)-\LL_1).
\end{align*}
We notice that the right operator $\scal{x_1}^{2\nu}(H_\omega+\Theta)R_{\omega,a}(\LL_{1,a}^\omega(t)-\LL_1)$ corresponds to the right operator treated 
in \eqref{lim 1 1}, while the left one 
$$R_{\omega,a}\com{R_{\omega,a},\com{H_{\omega,a},\LL_2}}(H_\omega+\Theta)\scal{x_2}^{2\nu}$$ 
has been treated in \eqref{rr 12}.

\subsubsection{Proof of Lemma~\ref{average}}\label{avr}
According to the spectral theorem and the assumption on $g$, 
we have
\begin{equation}\label{sp}
g(H_\omega)= \int g(E) \ \dd\p(E)= -\int g'(E) \p \ \dd E, 
\end{equation}
since $g(+\infty)P_\omega^{(+\infty)} -g(-\infty)P_\omega^{(-\infty)} = 0$. 
Thanks to \eqref{sp} we work with the Fermi projection $\p$ and we are left with the study of $\com{\p,\LL_2}(\LL_1^\omega(t)-\LL_1)$.
\vspace{0.3cm}
\newline
In the first step, we show that the operator $\com{\p,\LL_2}(\LL_1^\omega(t)-\LL_1)$ is trace class uniformly in $t$. Using the Duhamel expansion \ref{intgr}, 
it is enough to prove that the operator 
\begin{equation}\label{p lambda}
 \com{\p,\LL_2}\ \e^{isH_\omega}\com{H_\omega,\LL_1}\e^{isH_\omega},
\end{equation}
is trace class for $0\le s\le t$. Notice that 
\begin{equation}\label{p com com p}
 \com{\p,\LL_2}=\p\com{\p,\LL_2} + \com{\p,\LL_2}\p.
\end{equation}
We introduce $(H_\omega-\Theta+1)R_\omega(1-\Theta) $ inside \eqref{p lambda} such that
\begin{align}
 \eqref{p lambda}
&=\p\com{\p,\LL_2}(H_\omega+\Theta-1)\ \e^{isH_\omega} R_\omega(1-\Theta)\com{H_\omega,\LL_1}\e^{isH_\omega}\label{p lambda 1}\\
&+ \com{\p,\LL_2}\p(H_\omega+\Theta-1)\ \e^{isH_\omega} R_\omega(1-\Theta)\com{H_\omega,\LL_1}\e^{isH_\omega}\label{p lambda 2}.
\end{align}
We start with the term \eqref{p lambda 2} that we rewrote as
\begin{align}
 \com{\p,\LL_2}\e^{|x_2|^\zeta} &\left(\e^{-|x_2|^\zeta}\p(H_\omega+\Theta-1)\e^{-|x_1|^\zeta}\right)\notag\\
&\e^{|x_1|^\zeta}\e^{isH_\omega} R_\omega(1-\Theta)\com{H_\omega,\LL_1}\e^{isH_\omega}.\label{p lambda 2 1}
\end{align}
Since the operator $\e^{-|x_2|^\zeta}\p(H_\omega+\Theta-1)\e^{-|x_1|^\zeta}$ is well localized in energy and space, it is trace class. 
Moreover, the left and right operators in \ref{p lambda 2 1} are bounded by Lemma~\ref{resolvent} and Lemma~\ref{dec com}.
\vspace{0.3cm}
\newline
We come back now to \eqref{p lambda 1} and use that 
$$\com{\p,\LL_2}(H_\omega+\Theta-1)=\com{\p(H_\omega+\Theta-1),\LL_2}-\p\com{H_\omega,\LL_2},$$
to write
\begin{align}
 \eqref{p lambda 1}
&=\p\com{\p(H_\omega+\Theta-1),\LL_2}\e^{isH_\omega} R_\omega(1-\Theta)\com{H_\omega,\LL_1}\e^{isH_\omega}\label{p lambda 1 1}\\
&-\p\com{H_\omega,\LL_2}\e^{isH_\omega} R_\omega(1-\Theta)\com{H_\omega,\LL_1}\e^{isH_\omega}\label{p lambda 1 2}.
\end{align}
We expand these terms \eqref{p lambda 1 1} and \eqref{p lambda 1 2} as the sums of
\begin{equation}\label{p lambda 1 1'}
 \p\chi_x\com{\p(H_\omega+\Theta-1),\LL_2}\chi_y\ \e^{isH_\omega} R_\omega(1-\Theta)\com{H_\omega,\LL_1}\chi_u\ \e^{isH_\omega}
\end{equation}
and
\begin{equation}\label{p lambda 1 2'}
-\p(H_\omega+\Theta-1)\chi_x R_\omega(1-\Theta)\com{H_\omega,\LL_2}\chi_y \ \e^{isH_\omega} R_\omega(1-\Theta)\com{H_\omega,\LL_1}\chi_u\ \e^{isH_\omega}
\end{equation}
respectively over $\z^2\times\{\z\times(\z\cap[-1,1])\}\times\{(\z\cap[-1,1])\times\z\}$.
Since 
$$\sup_x\trnorm{\p\chi_x}<\infty \ \ \mathrm{and}\ \ \sup_x\trnorm{\p(H_\omega+\Theta-1)\chi_x}<\infty,$$ 
we use Lemma~\ref{resolvent} to obtain an exponential decay of the kernels
$$\chi_y \ \e^{isH_\omega} R_\omega(1-\Theta)\com{H_\omega,\LL_1}\chi_u \ \ \mathrm{and} \ \ \chi_x R_\omega(1-\Theta)\com{H_\omega,\LL_2}\chi_y,$$ 
in operator norm to deduce
the summability of \eqref{p lambda 1 1'} and \eqref{p lambda 1 2'}. Therefore, the operator 
$\com{\p,\LL_2}(\LL_1^\omega(t)-\LL_1)$ is trace class.
In the next step, we consider the decomposition 
\begin{equation}\label{dec}
 \com{\p,\LL_2}=\com{\p,\LL_2}\pp +\com{\p,\LL_2}\p
=\p\LL_2 \pp - \pp \LL_2 \p 
\end{equation}
and we write
\begin{align}\label{dec 1}
\com{\p,\LL_2}(\LL_1^\omega(t)-\LL_1)
= \p \LL_2 & \pp(\LL_1^\omega(t)-\LL_1)\notag\\
&- \pp \LL_2 \p(\LL_1^\omega(t)-\LL_1).
\end{align}
Both operators on the r.h.s of \eqref{dec 1} are separately trace class. Hence, we can cycle the projections $\p$ and $\pp$ around the trace of \eqref{dec 1}.
Setting
\begin{equation}\label{Pi_E}
 \Pi_E:=\p \LL_2 \pp\LL_1 \p -\pp \LL_2 \p \LL_1 \pp,
\end{equation}
and 
\begin{equation}\label{Pi_E(T)}
\Pi_E(t) := \p \LL_2 \pp \LL_1^\omega(t) \p-\pp \LL_2 \p \LL_1^\omega(t) \pp, 
\end{equation}
one gets
\begin{equation}
\tr\eqref{dec 1}=\tr \Pi_E(t) -\tr \Pi_E .\label{Pi T}
\end{equation}
We claim that the time-average of the trace of $\Pi_E(t) $ vanishes as $T$ tends to $\infty$. 
Indeed, we rewrite
\begin{equation}
\frac{1}{T}\int_0^T\Pi_E(t)\ \dd t= \frac{1}{T}\int_0^T (\p \LL_2 \pp \LL_1^\omega(t) \p-\pp \LL_2 \p \LL_1^\omega(t) \pp) \ \dd t,
\end{equation}
as the sum of
\begin{align}
 \int_{\substack{\lambda>E \\ \mu\leq E}} \left(\frac{1}{T} \int_0^T \e^{-it(\mu-\lambda)} \dd t \right) \p \LL_2 \pp \dd P_\lambda \ \LL_1 \ \dd P_\mu \p,
\end{align}
and 
\begin{align}
\int_{\substack{\lambda\leq E \\ \mu> E}} \left(\frac{1}{T} \int_0^T \e^{-it(\mu-\lambda)} \dd t \right)
&\pp \LL_2 \p \dd P_\lambda \ \LL_1 \ \dd P_\mu \pp.
\end{align}
Since $\lambda\neq\mu$ and $|\frac{\e^{ix}-1}{x}|\leq 1$, 
we have $$\frac{1}{T}\int_0^T \e^{-it(\mu-\lambda)} \dd t=\frac{e^{-iT(\mu-\lambda)}-1}{-iT(\mu-\lambda)}\to 0,$$ 
when $T$ tends to $\infty$. 
Using the theorem of dominated convergence we complete the proof.

\subsection{Bulk-Edge equality}

\begin{proof}[Proof of Lemma~\ref{dec Hall}]
 We decompose the commutator within the bulk conductance \eqref{Hall conductane} and we insert $ - \p\LL_2 \ \LL_1 \p + \p \LL_1 \ \LL_2 \p$.
Using that $\LL_1\LL_2=\LL_2\LL_1$, one obtains  
\begin{align} \sigma_{\mathrm{Hall}}(B,\omega,E)
& =-i\tr \com{\p\LL_2 \p , \p\LL_1 \p}\notag\\
&=i\tr(\p \LL_2 \pp \LL_1 \p - \p \LL_1 \pp \LL_2 \p).
\end{align}
Moreover, to see that 
\begin{equation}\label{Lambda 1 2-2 1}
 \tr( \p \LL_1 \pp \LL_2 \p) = \tr (\pp \LL_2 \p \LL_1 \pp),
\end{equation}
we apply Proposition~\ref{si} and for instance we write
\begin{equation}
 \p \LL_1 \pp \LL_2 \p= \com{\p,\LL_1} \pp \com{\LL_2,\p},
\end{equation}
which is seen to be trace class by cyclicity and Lemma~\ref{dec com}. The same argument works for
$\pp \LL_2 \p \LL_1 \pp$. We thus get
$$\tr( \p \LL_1 \pp \LL_2 \p)= \tr(\pp \LL_2 \p \LL_1 \pp).$$
Recalling that
\begin{equation}\label{Pi_E}
\Pi_E = \p \LL_2 \pp \LL_1 \p - \pp \LL_2 \p \LL_1 \pp,
\end{equation}
 one has $ \sigma_\mathrm{Hall}(E) = i\tr \Pi_E$ and \eqref{dec Hall cond} follows. 
\end{proof}
\bigskip

Theorem~\ref{equality} is derived from the analysis done in the previous sections.
\begin{proof}[Proof of Theorem~\ref{equality}.]
Combining the fact that 
$\sigma_\mathrm{Hall}(E) =i\tr\Pi_E $  and the constancy of Hall conductance $\sigma_H$ in connexe intervals of localization, we conclude that 
\begin{equation*}
 \sigma_{e,\omega}^{\mathrm{reg}}=-\int g'(E) \ \sigma_\mathrm{Hall}(E)\ \dd E = \sigma_\mathrm{Hall}.
\end{equation*}

\end{proof}


\appendix
\section{Technical tools}\label{A}

\begin{lemma}\cite{Si}\label{stron lim}
 Let $A_n \in\B$ such that $A_n\xrightarrow s A$ and let $B \in\T_p$ for $p>0$. Then we have $\pnorm{A_n B-AB}\to 0 $.
\end{lemma}

\begin{proof}
Since
\begin{equation*}
 \T_p=\overline{(\mathrm{Finite \ rank \ operators})}_{\pnorm{.}}
\end{equation*}
there exists a finite rank operator $P$ such that $\pnorm{(1-P)B}\le\eps $ for a given $\eps>0$.
Write
\begin{align}
\pnorm{(A_n-A)B}
&=\pnorm{(A_n-A)(B-PB+PB)}\notag\\
&\le \norm{(A_n-A)P}\pnorm{B}+\norm{(A_n-A)}\pnorm{(1-P)B}\notag\\
&\le\eps(\norm{A_n}+\norm{A}+\pnorm{B})\notag
\end{align}
where we have used that by strong convergence we have $(A_n-A)P\to0$ and the result holds since $\eps$ is arbitrarily chosen.
\end{proof}

\begin{proposition}\label{si}\cite{Si}\mbox{}
\begin{enumerate}
\item[(i)] Let $A_n\xrightarrow s A$ and $B$ be a compact operator. Then $\norm{A_nB-AB}\to0 $. 
\item[(ii)] Let $A,B\in\B$. If $AB, BA\in\T_1$ then $\tr AB=\tr BA$.
\item[(iii)] Let $B\in\B$ and $A\in\T_1$. Then we have $\tr AB = \tr BA$. 
\item[(iv)] Let $A_n,B_n\in\B$ such that $A_n\xrightarrow s A$ and $B_n\xrightarrow s B$. Then $A_nB_n\xrightarrow s AB$.
\end{enumerate}
\end{proposition}

Next, we reproduce \cite[Lemma 3]{CG} that we adapt to obtain a decay of the kernel $\chi_x \ \e^{-itH} R(z)\com{H,\LL_2}\chi_y$ in operator norm.
\begin{lemma}\label{resolvent}
Let $\chi_x$ and  $\chi_y$ be two smooth functions.
Let $R_A(z)$  be the resolvent of the operator $H(A)= (-i\nabla-A)^2$. Then there exist $c>0$ and $C_t$ such that
\begin{equation}\label{CT t}
  \norm{\chi_x \ \e^{-itH(A)} R_A(z) \com{H(A),\Lambda_2} \chi_y} \leq \frac{C_t}{\eta} \ \e^{-c \eta(|x_1-y_1|+|x_2-y_2|)}
 \end{equation}
for all $z\notin\sigma(H(A))$ and $x,y\in\R^d$ and where $\eta=\mathrm{dist}(z,\sigma(H))$.
\end{lemma}

\begin{proof}
We follow the same procedure used in \cite[Lemma 3]{CG}. We consider the vector potential $A=(0,\beta(x_1))$ 
and we let $\tilde\chi_j$ smooth functions with $\tilde \chi_j=1$ on $\mathrm{supp} \ \chi_j$ for $j=x,y$.
We take $y_2\in\mathrm{supp} \LL_2'$ otherwise \eqref{CT t} is equal to zero. 
\bigskip

We write 
$H(A)= (-i\nabla-A)^2 = \Pi_1 ^2 +\Pi_2^2 $ where $\Pi_1=p_1 $ and $\Pi_2=p_2- \beta(x_1) $. 
Let us estimate the decay of $\chi_x\ \e^{-itH(A)}R_A(z)\com{H,\Lambda_2} \chi_y $ for $t\in\R$.
Notice that
\begin{align}
 \com{H(A),\Lambda_2}
&=-i(-i\nabla-A).\nabla\Lambda_2 - i\nabla\Lambda_2.(-i \nabla-A)\notag\\
&=-i\Pi_2  \Lambda_2' -i\Lambda_2' \ \Pi_2\notag\\
&=-\Lambda_2''-2i\Lambda_2' \ \Pi_2,\notag
\end{align}
and since for $\varphi\in\mathcal{C}_0^\infty(\R^2)$, we have
\begin{equation*}
||\chi_x \e^{-itH(A)} R_A(z)\Pi_2\chi_y \varphi||^2
=\scal{\chi_y\ \Pi_2R_A(\overline{z}) \ \e^{itH(A)}\chi_x^2 \ \e^{-itH(A)} R_A(z)\Pi_2\chi_y \varphi,\varphi},
\end{equation*}
it is enough to bound
$\norm{\chi_y \ \Pi_2 R_A(\overline{z})\ \e^{itH(A)}\chi_x} $. 
We write 
\begin{align}
||\chi_y\ \Pi_2 & R_A(\overline{z})\ \e^{itH(A)}\chi_x\varphi||^2
=\scal{R_A(\overline{z})\ \e^{itH(A)}\chi_x\varphi,\Pi_2 \chi_y^2 \ \Pi_2 R_A(\overline{z}) \ \e^{itH(A)}
\chi_x\varphi}\notag\\
&= \scal{R_A(\overline{z})\ \e^{itH(A)}\chi_x\varphi,\tilde\chi_y(2(p_2\chi_y)+\beta(x_1)\chi_y)\chi_y 
\ \Pi_2 R_A(\overline{z})\ \e^{itH(A)}\chi_x\varphi}\label{(I)}\\
&+2\scal{R_A(\overline{z})\ \e^{itH(A)}\chi_x\varphi,\chi_y^2 \ \Pi_2^2 R_A(\overline{z})\ \e^{itH(A)}\chi_1\varphi},\label{(II)}
\end{align}
where we used 
$$\Pi_2 \chi_y^2 \Pi_2= (p_2\chi_y^2)\Pi_2 + \chi_y^2 \Pi_2^2 = 2(p_2 \chi_y)\chi_y \Pi_2 + \chi_y^2 \beta(x_1)\Pi_2 +2\chi_y^2\Pi_2^2$$
together with 
$$p_{2}\chi_y^2= 2(p_{2} \chi_y)\chi_y + \chi_y^2 \ p_{2}=\chi_y^2 \ \Pi_{2} +\tilde\chi_y(2(p_{2}\chi_y)+\beta(x_1)\chi_y)\chi_y.$$
Similarly, we have $$ \Pi_{1} \chi_y^2\ \Pi_{1} =(\Pi_{1}\chi_y^2)\ \Pi_{1}+ \chi_y^2\ \Pi_{1}^2$$
and $$\Pi_{1}\chi_y^2 = 2\tilde\chi_y((p_{2}\chi_y)\chi_y + \chi_y^2 \ \Pi_{1}.$$
Hence
\begin{align}
 ||\chi_y\ \Pi_{1}R_A(\overline{z})&\ \e^{itH(A)}\chi_x\varphi||^2
=\scal{R_A(\overline{z})\ \e^{itH(A)}\chi_x\varphi,\Pi_{1} \chi_y^2 \ \Pi_{1} R_A(\overline{z}) \ \e^{itH(A)}
\chi_x\varphi}\notag\\
&= \scal{R_A(\overline{z})\ \e^{itH(A)}\chi_x\varphi,\tilde\chi_y(2(p_{2}\chi_y)\chi_y 
\ \Pi_{2} R_A(\overline{z})\ \e^{itH(A)}\chi_x\varphi}\label{(I')}\\
&+\scal{R_A(\overline{z})\ \e^{itH(A)}\chi_x\varphi,\chi_y^2 \ \Pi_{2}^2 R_A(\overline{z})\ \e^{itH(A)}\chi_x\varphi}.\label{(II')}
\end{align}
We first estimate \eqref{(I)} so that
\begin{align}
|\eqref{(I)}|
&\le\norm{2(p_{2}\chi_y)+\beta(x_1)\chi_y}_\infty \norm{\tilde\chi_y R_A(\overline{z})\ \e^{itH(A)}\chi_x\varphi}
\norm{\chi_x \ \Pi_{2} R_A(\overline{z})\ \e^{itH(A)}\chi_x\varphi}\notag\\
&\le \frac{1}{2}\norm{2(p_{2}\chi_y)+\beta(x_1)\chi_y)}_\infty^2 \norm{\tilde\chi_y R_A(\overline{z})\ \e^{itH(A)}\chi_x\varphi}^2
+\frac{1}{2}\norm{\chi_y \ \Pi_{2} R_A(\overline{z})\ \e^{itH(A)}\chi_x\varphi}^2.\notag
\end{align}
In the same manner, one has
\begin{align}
|\eqref{(I')}|
&\le\norm{2(p_{1}\chi_y)}_\infty \norm{\tilde\chi_y R_A(\overline{z})\ \e^{itH(A)}\chi_x\varphi}
\norm{\chi_y \ \Pi_{1} R_A(\overline{z})\ \e^{itH(A)}\chi_x\varphi}\notag\\
&\le 2\norm{(p_{1}\chi_y)}_\infty^2 \norm{\tilde\chi_y R_A(\overline{z})\ \e^{itH(A)}\chi_x\varphi}^2
+\frac{1}{2}\norm{\chi_y \ \Pi_{1} R_A(\overline{z})\ \e^{itH(A)}\chi_x\varphi}^2.\notag
\end{align}
Therefore,
\begin{align}
||\chi_y\ \Pi_{1}R_A(\overline{z})&\ \e^{itH(A)}\chi_x\varphi||^2 + ||\chi_y\ \Pi_{2}R_A(\overline{z})\ \e^{itH(A)}\chi_x||^2
\notag\\
&\le (\norm{2(p_{2}\chi_y) + \beta(x_1)\chi_y}_\infty^2 +4\norm{(p_{1}\chi_y)}_\infty^2)\norm{\tilde\chi_y R_A(\overline{z})\ \e^{itH(A)}
\chi_x\varphi}^2\notag\\
&+2|\scal{R_A(\overline{z})\ \e^{itH(A)}\chi_x\varphi,\chi_y^2 \ (\Pi_{1}^2 + \Pi_{2}^2)\ R_A(\overline{z})\ \e^{itH(A)}\chi_x\varphi}| 
\notag\\
&\le (4 \norm{(p_{1}\chi_y)}_\infty^2 + 6\norm{(p_{2}\chi_y)}_\infty^2 +2\norm{\beta(x_1)\chi_y}_\infty^2)
\norm{\tilde\chi_y R_A(\overline{z})\ \e^{itH(A)} \chi_x\varphi}^2\notag\\
&+ 2\norm{\chi_y R_A(\overline{z})\ \e^{itH(A)}\chi_x \varphi} \norm{\chi_y\chi_x\varphi} 
+ 2|\overline{z}|\norm{\chi_y R_A(\overline{z})\ \e^{itH(A)}\chi_x\varphi}^2,
\end{align}
where we used that $(\Pi_{1}^2 + \Pi_{2}^2 ) R_A(\overline{z})= I + \overline{z} R_A(\overline{z})$.
Now since the function $\frac{\e^{-itu}}{u-z}$ is analytic for $\im z\neq 0 $ (then it is of Gevrey class), it follows from \cite{GK3,BGK}
\begin{equation}
 \norm{\chi_x \ \e^{-itH(A)} R_A(z) \chi_y}\le \frac{\e^{c_1 t}}{\eta} \ \e^{-c_2\eta |x-y|},
\end{equation}
and the lemme holds.
\end{proof}
\bigskip

The following lemma establishes the decay of the kernel operator of $\com{\p,\Lambda_2}$. 
\begin{lemma}\label{dec com}
Assume \eqref{DFP_omega}. Then we have
\begin{equation}\label{decay_commutator}
 \normsch{\chi_x\com{P_\omega^{(E)}, \Lambda_2}\chi_y}\le C_{\omega,m,\zeta,\eps,B,E} \ \e^{\eps|x|^\zeta} \e^{-\frac{m}{2}|x_1-y_1|^\zeta - 
\frac{m}{4}|x_2|^\zeta -\frac{m}{4}|y_2|^\zeta},
\end{equation}
for all $x,y\in\z^2$.
Moreover, the operator $[\p,\LL_1]\pp[\p,\LL_2]$ is trace class.
\end{lemma}

\begin{proof}
We recall the definition of the function $\LL$. We considered $\LL(s)=1$ for $s\le-\frac12$ and $\LL(s)=0$ for $s\ge\frac12$ such that 
$\mathrm{supp}\LL_1'\subset(-\frac12,\frac12)$.

We expand the commutator such that we have 
\begin{equation}\label{kernel commutator}
  \chi_x \com{P_\omega^{(E)},\Lambda_2}\chi_y=\chi_x \ \p \LL_2 \chi_y-\chi_x \ \LL_2 \p \chi_y.
\end{equation}
If $x_2,y_2\ge1$ or $x_2,y_2\le -1$ then we have $\eqref{kernel commutator}=0$.
Consider now the case $y_2\le -1, x_2\ge 1 \ \mathrm{or}\ y_2\ge1, x_2\le -1$.
where we get $\eqref{kernel commutator}=\pm \chi_x P_\omega^{(E)} \chi_y $.
Thus the decay can be obtained by \eqref{DFP_omega} and a use of
\begin{equation}
 \e^{-|x-y|^\zeta} \le \e^{-\frac{1}{2}|x_1-y_1|^\zeta -\frac{1}{2}|x_2-y_2|^\zeta},\notag
\end{equation}
and the fact that in the present case, we have $|x_2-y_2|^\zeta=(|x_2|+|y_2|)^\zeta\ge \frac{1}{2}|x_1|^\zeta + \frac{1}{2}|x_2|^\zeta$.
The case of $x_2=0$ or $y_2=0$ yields \eqref{decay_commutator} since it follows from \eqref{DFP_omega} for instance for $x_2=0$ that  

\begin{equation}
 \normsch{\chi_x \LL_2 \p\chi_y}\le C_{\omega,m,\zeta,B,E}\ \e^{\eps|x_1|^\zeta} \e^{-m|x_1-y_1|^\zeta -m|y_2|^\zeta}.
\end{equation}

Moreover, it follows from \eqref{decay_commutator} that the operator $[\p,\LL_1]\pp[\p,\LL_2]$ is trace class.
\end{proof}
\bigskip

\begin{lemma}\label{A.5}

Let $R$ be the resolvent of the operator $H$. Then the operators
$$\chi_x R^2\com{H,\Lambda_j}, \ R^2\com{H,\Lambda_j}\chi_x ,\ R\com{H,\Lambda_2}R \ \chi_x \ \in\T_1.$$
\end{lemma}

\begin{proof}
Let $M<\inf \sigma(H)$. We introduce the power resolvent $R^2(M)$ and we write
\begin{equation}
 \chi_x R^2(z)\com{H,\LL_j}=\chi_x R^{\frac{3}{2}}(M) R^2(z)(H+M)^2 R^{\frac{1}{2}}(M)\com{H,\LL_j}.
\end{equation}
The trace class property follows from the fact that $R^2(z)(H+M)^2$ and $R^{\frac{1}{2}}(M)\com{H,\LL_j}$ are bounded and $\chi_x R^{\frac{3}{2}}(M)$ is trace class.
In particular, $ R^2(z)\com{H,\LL_j}\chi_x$ is also trace class since an operator $T$ belongs to $\T_1$ if and only if $T^* \in\T_1$.

\bigskip
For the last operator, it suffices to see that

\begin{equation}
 R\com{H,\Lambda_j} R\ \chi_x
=\com{H,\Lambda_j} R^2 \chi_x - R\com{H,\com{H,\Lambda_j}} R^2 \chi_x,
\end{equation}
and since $ \com{H,\Lambda_j} R^{\frac{1}{2}}$ is bounded as well as $R\com{H,\com{H,\Lambda_j}} $
and both $R^{\frac{3}{2}} \chi_x $ and $R^2 \chi_x $ are trace class then the operator $ R\com{H,\Lambda_j} R\ \chi_x$ is trace class.  
 
\end{proof}


\section{Strong convergence}\label{s-cv}

\begin{lemma}\label{strong cv R}
Let $H_\omega$ and $H_{\omega,a}$ be the operator defined in the sections~\ref{Bulk} and \ref{Edge} with corresponding resolvents $R_{\omega}$ and $R_{\omega,a}$. Then 
\begin{equation}\label{cv R}
 R_{\omega,a} \xrightarrow s R_{\omega},
\end{equation}
as $a\to\infty$ and $ \mathbb{P}-\mathrm{a.e}\ \omega$.
In particular, one has
\begin{equation}\label{cv lambda}
 \LL_{1,a}^\omega(t) \xrightarrow s \LL_1^\omega(t) \ \ \mathbb{P}-\mathrm{a.e}\ \omega,
\end{equation}
for all $t\in\R$.
\end{lemma}

\begin{proof}
We have $ R_{\omega,a}-R_\omega =-R_{\omega}\Gamma_{\omega,a}R_{\omega,a} $ and since $R_\omega$ is bounded, 
and $\Gamma_{\omega,a}R_{\omega,a}$ is uniformly bounded in $a$,
it suffices to prove the strong convergence of $\Gamma_{\omega,a}R_{\omega,a}$ in $\mathcal{C}_0^\infty(\R^2)$.
Let $f\in\mathcal{C}_0^\infty(\R^2)$ such that $\mathrm{supp}f=D_{r_1,r_2}=[-r_1,r_1]\times[-r_2,r_2]$ with $r_1>a$ and $r_2>0$.
\bigskip

$\bullet$ \textit{Electric case}. In this case, one has
\begin{equation}\label{st cv elec}
 \norm{\Gamma_{\omega,a}^\ee R_{\omega,a}^\ee f}= \norm{U_a R_{\omega,a}^\ee f}\le c_1 |\im z|^{-1} \ \e^{-\tilde c_1 |\im z||a-r_1|}
\end{equation}
which goes to $0$ as $a\to+\infty$ and where we have used Combes-Thomas estimate \cite{CT,GK1}.

$\bullet$ \textit{Magnetic case}. 
Since the magnetic field $\B$ is basically generated in the region $\mathcal{P}_a:=(\infty,-a)\times\R$, 
it follows from \cite[Proposition 4.2]{DGR1} that we apply for this semi-plane $\mathcal{P}_a$, that the vector potential vanishes outside $\mathcal{P}_a$.
This means that the operator
$$\Gamma_{\omega,a}^\m=-2\A_a.(-i\nabla-\A_0-\A_\omega)+i\dv\A_a+|\A_a|^2 $$
is supported on $\mathcal{P}_a$ and the strong convergence of $\Gamma_{\omega,a}^\m R_{\omega,a}^\m f$ follows similarly to \eqref{st cv elec}.
The second point \eqref{cv lambda} is a consequence of \eqref{cv R}, \cite{RS}.
\end{proof}
\bigskip

\begin{acknowledgement}
The author is grateful to F. Germinet for his valuable support and several discussions.
\end{acknowledgement}

\end{document}